\documentclass[lettersize, journal]{IEEEtran}
\usepackage{setspace}
\usepackage{cite}
\usepackage{pgfplots}

\pgfplotsset{compat=1.10}

\usepgfplotslibrary{fillbetween}

\usepackage{needspace}
\usepackage{balance}
\usepackage{bm}
\usepackage{lipsum}
\usepackage{makeidx}
\usepackage{enumerate}
\usepackage{color}
\usepackage{cite}
\usepackage{amsmath,amsthm}   
\usepackage{amssymb}
\usepackage{nomencl}
\usepackage{multirow}
\usepackage{graphicx}
\usepackage{epstopdf}
\usepackage{threeparttable}
\usepackage{multicol}
\usepackage{tikz}
\usetikzlibrary{matrix} 
\usetikzlibrary{positioning}
\DeclareGraphicsExtensions{.pdf,.jpeg,.png,.jpg,.emf,.eps}
\usepackage{algpseudocode}
\usepackage{algorithmicx}
\usepackage[ruled,noline,linesnumbered]{algorithm2e}
\usepackage{listings}
\usepackage[numbered,framed]{matlab-prettifier} 

\usepackage{calc}
\usepackage{here}
\usepackage{hyperref}
\usepackage{xurl}

\usepackage{upref}
\usepackage{comment}
\usepackage{times}
\usepackage{dsfont}
\usepackage{epic,eepic}
\usepackage{rawfonts}
\usepackage[T1]{fontenc}
\usepackage{latexsym}
\usepackage{amsfonts}

\usepackage{capitalgreekitalic}

\usepackage{xcolor}
\usepackage{tikz,pgfplots}
\usetikzlibrary{calc}
\usetikzlibrary{intersections}
\usetikzlibrary{arrows,shapes}
\usetikzlibrary{positioning}

\newtheorem{theorem}{Theorem}

\newtheorem{lemma}{Lemma}

\newtheorem{remark}{Remark}

\theoremstyle{definition}
\newtheorem{definition}{Definition}

\def\0{{\bf 0}}
\def\Y{{\bf Y}}
\def\X{{\bf X}}
\def\D{{\bf D}}
\def\R{{\bf R}}
\def\W{{\bf W}}

\def\Q{{\bf Q}}
\def\I{{\bf I}}
\def\A{{\bf A}}

\def\U{{\bf U}}
\def\V{{\bf V}}
\def\F{{\bf F}}
\def\T{{\bf T}}
\def\G{{\bf G}}
\def\H{{\bf H}}
\def\U{{\bf U}}
\def\B{{\bf B}}
\def\C{{\bf C}}

\def\x{{\bf x}}
\def\y{{\bf y}}

\def\b{{\bf b}}

\def\f{{\bf f}}
\def\g{{\bf g}}
\def\u{{\bf u}}

\def\Thetab{\bm{\Theta}}
\def\Sigmab{\bm{\Sigma}}
\def\Lambdab{\bm{\Lambda}}

\def\Gammab{\bm{\Gamma}}

\def\sigmab{\bm{\sigma}}
\def\tr{\operatorname{tr}}
\def\det{\operatorname{det}}
\def\diag{\operatorname{diag}}
\def\blkdiag{\operatorname{blkdiag}}
\def\Re{\operatorname{Re}}
\def\Im{\operatorname{Im}}
\def\rank{\operatorname{rank}}
\def\vec{\operatorname{vec}}

\newcommand{\CU}{\mathcal{U}}

\DeclareMathOperator*{\argmin}{arg\,min}
\allowdisplaybreaks
\begin{document}

\title{{Optimal symmetric low-rank BD-RIS configuration maximizing the determinant of a MIMO link}

\author{Ignacio~Santamaria, \IEEEmembership{ Senior Member, IEEE}, Mohammad Soleymani, \IEEEmembership{Senior Member, IEEE}, Jes{\'u}s Guti{\'e}rrez, \IEEEmembership{Member, IEEE}, Eduard Jorswieck, \IEEEmembership{Fellow, IEEE}}
\thanks{I. Santamaria is with the Department
of Communications Engineering, Universidad de Cantabria, 39005 Santander, Spain (e-mail: i.santamaria@unican.es).}
\thanks{M. Soleymani is with Department of Electronic and Electrical Engineering, Brunel University London, Uxbridge, U.K. (e-mail: mohammad.soleymani@brunel.ac.uk).}
\thanks{J. Guti{\'e}rrez is with IHP - Leibniz-Institut
f{\"u}r Innovative Mikroelektronik, 15236 Frankfurt (Oder), Germany (email: teran@ihp-microelectronics.com).}
\thanks{Eduard Jorswieck is with Institute for Communications Technology, Technische Universit{\"a}t Braunschweig, 38106 Braunschweig,
Germany (email: e.jorswieck@tu-bs.de).}
\thanks{This work is supported by the European Commission’s Horizon Europe, Smart Networks and Services Joint Undertaking, research and innovation program under grant agreement 101139282, 6G-SENSES project. The work of I. Santamaria was also partly supported under grant PID2022-137099NB-C43 (MADDIE) funded by MICIU/AEI /10.13039/501100011033 and FEDER, UE. The work of E. Jorswieck was carried out in part within the Nitride Technology Center’s NTC BRIGHT initiative (Grant No. ZN4786), funded by zukunft.niedersachsen, the joint science funding program of the Lower Saxony Ministry of Science and Culture and the Volkswagen Foundation.}
}

\maketitle

\begin{abstract}

Beyond-diagonal reconfigurable intelligent surfaces (BD-RISs) significantly improve wireless performance by allowing tunable interconnections among elements, but their design in multiple-input multiple-output (MIMO) systems has so far relied on complex iterative algorithms or suboptimal approximations.  This work introduces a simple yet powerful approach: instead of directly maximizing the achievable rate, we maximize the absolute value of the determinant of the equivalent MIMO channel. We derive a closed-form symmetric passive scattering matrix whose rank is exactly twice the channel’s degrees of freedom (\(2r\)). Remarkably, this low-rank solution achieves the same determinant value as the optimal unitary BD-RIS. Using log-majorization theory, we prove that the rate loss relative to the optimal unitary BD-RIS vanishes at high signal-to-noise ratio (SNR) or when the number of BD-RIS elements becomes large. Moreover, the proposed solution can be perfectly implemented using a $q$-stem BD-RIS architecture with only \(q = 2r-1\) stems, requiring a minimum number of reconfigurable circuits. The resulting Max-Det solution is orders of magnitude faster to compute than existing iterative methods while achieving near-optimal rates in practical scenarios. This makes high-performance BD-RIS deployment feasible even with large surfaces and limited computational resources.  

\end{abstract}

\section{Introduction}
Beyond-Diagonal Reconfigurable Intelligent Surfaces (BD-RISs), originally proposed in \cite{ClerckxTWC22a}, \cite{ClerckxTWC22b}, are receiving a lot of attention recently as they provide improved control of the amplitudes and phases of the reflective elements, thus providing greater flexibility for channel shaping than diagonal RIS. A comprehensive and up-to-date review of the modeling, optimization, and application of BD-RIS can be found in \cite{ClerkxBDRISTut26}. 

BD-RIS are usually modeled as passive and reciprocal devices. For a BD-RIS with $M$ elements, this means that the $M \times M$ scattering matrix, $\Thetab$, all of whose entries can be optimized, must satisfy $\Thetab \Thetab^H \preceq \I_M$, to be a passive device (with equality in the case of a lossless BD-RIS), and $\Thetab = \Thetab^T$ due to reciprocity. For single-stream transmissions, closed-form BD-RIS solutions have been found for certain optimization criteria. For example, \cite{SantamariaSPLetters2023} shows that the Takagi decomposition \cite{Takagi} \cite[Sec. 4.4]{HornBook} of the outer product of forward channel, from the transmitter (Tx) to BD-RIS, and backward channel, from BD-RIS to the receiver (Rx), provides a closed-form solution for the unitary and symmetric BD-RIS that maximizes the signal-to-noise ratio (SNR), and, consequently, the capacity, in a single-input single-output (SISO) channel. The result can be easily extended to single-stream transmissions over single-input multiple-output (SIMO) or multiple-input single-output (MISO) channels. Exploiting a different factorization, closed-form BD-RIS solutions in single-stream scenarios have also been described in \cite{NeriniTWC2023}.

However, considering multiple-input multiple-output (MIMO) systems involves a higher level of difficulty, such that, to date, there are no known unitary and symmetric closed-form solutions for BD-RIS that optimize some of the usual metrics, such as the rate or the Frobenius norm of the equivalent channel. Existing solutions in MIMO scenarios are either iterative, with varying degrees of computational complexity, or suboptimal, based on some relaxation of the BD-RIS constraints followed by a projection onto the set of unitary and symmetric matrices. Representative examples of such approaches have been reported in the literature.

Specifically, a non-iterative but suboptimal solution for maximizing the Frobenius norm of the equivalent MIMO channel has been proposed in \cite{MaoCL2024}. Furthermore, this suboptimal solution requires an unblocked direct link between Tx and Rx. Focusing on rate maximization, which is the metric of interest in this work, an iterative solution for BD-RIS-assisted MIMO links was proposed in \cite{SantamariaSPAWC24}. This solution, which exploits Takagi decomposition and applies manifold optimization (MO) techniques, provides a stationary point of the cost function but with a very high computational cost. More recently, a less complex MO algorithm that operates directly on the manifold of unitary and symmetric matrices has been proposed in \cite{santamariaICASSP26}. However, its complexity is still excessive for practical applications if the number of elements in the BD-RIS is large. Similar comments can be made about the generic BD-RIS optimization framework proposed in \cite{Zhou24}, which is based on the penalty dual decomposition methodology.

Beyond algorithmic considerations, the hardware design complexity of BD-RIS is another practical concern, which has attracted increasing attention. While fully-connected BD-RIS architectures can provide notable performance gains over conventional diagonal RIS, they require a large number of tunable circuit elements, resulting in high implementation complexity, which may affect the energy efficiency of BD-RIS \cite{SoleymaniSEE26,You21, SoleymaniSPAWC24}. To mitigate this issue, several reduced-complexity interconnection topologies have been proposed, including group-connected, tree-connected, and stem-connected designs \cite{ClerkxBDRISTut26}. These architectures substantially lower circuit complexity while preserving the key advantage of BD-RIS, namely, controllable inter-element coupling. In particular, stem-connected structures have been shown to match the performance of fully-connected implementations with significantly fewer hardware components in multi-user MIMO systems \cite{WuTIT25}. Such designs therefore offer an attractive performance-complexity trade-off, making BD-RIS more viable for practical deployment.

Perfect CSI is adopted in this work to isolate the performance of the proposed transmission/optimization scheme from channel-estimation errors. This assumption is commonly employed as a benchmark in RIS-assisted optimization studies \cite{Wu19,Guo20}. CSI acquisition in passive RIS systems is nevertheless challenging and has been extensively studied using pilot-based training, appropriately designed RIS training configurations, element grouping, and sparsity-aware methods; see, e.g., \cite{Wang20,You20,Amri23} and the survey \cite{Zheng22}. For BD-RIS architectures, the non-diagonal scattering matrices introduce additional coupling between RIS elements and architecture-dependent constraints on the training configurations. Dedicated least-square (LS), tensor-, and structure-based estimation schemes have therefore recently been developed \cite{Li24,Almeida25,Wang25,Wang26}. In particular, \cite{Almeida25} provides decoupled estimates of the constituent channels, whereas \cite{Wang25,Wang26} show that structural dependencies can be exploited to substantially reduce the training overhead. Similar to \cite{Wu19,Guo20}, we assume that perfect CSI of the involved links is available in order to focus on the fundamental properties of the proposed Max-Det design. 

In this paper, we address both computational and implementation complexity challenges associated with BD-RIS. From a computational perspective, we consider the maximization of the absolute value of the equivalent MIMO channel as a proxy for rate and show that the resulting problem admits a closed-form solution. The proposed design is rate-optimal in the high-SNR regime when the direct Tx-Rx link is obstructed. Interestingly, the optimal scattering matrix is shown to have a rank exactly equal to twice the number of degrees of freedom (DoF) of the MIMO channel. From an implementation standpoint, this inherent low-rank structure enables efficient realization using a minimal number of impedance elements when the BD-RIS is implemented under the $q$-stem topology \cite{WuTIT25,ICC_QStem,Zhou2025Arxiv}. Furthermore, leveraging log-majorization analysis, we upper bound the rate loss of the proposed Max-Det solution and prove that it vanishes as either the SNR or the number of BD-RIS elements increases. Compared to manifold-optimization-based iterative designs, the proposed closed-form solution exhibits orders of magnitude lower computational complexity, making it a practical approach for rate maximization in BD-RIS-assisted MIMO systems.

The main contributions of this work can be summarized as follows:

\begin{itemize}

    \item For the first time, we consider the problem of maximizing the absolute value of the determinant of the equivalent MIMO channel assisted by a passive and symmetric BD-RIS, showing that maximizing the determinant is an asymptotically optimal proxy of rate as the SNR or the number of elements of the BD-RIS grows.

    \item We show that the problem of maximizing the determinant with a symmetric BD-RIS has a closed-form solution whose rank is twice the number of degrees of freedom (DoF) of the channel (i.e., $\rank = 2r$ with $r=\min(N_t,N_r)$). Compared to the optimal unitary BD-RIS, a unitary {\it and symmetric} solution simply doubles the rank of the solution. Both solutions yield equivalent MIMO channels, whose determinants have the same absolute value.

    \item An analysis based on log-majorization results is exploited to bound the rate gap between the Max-Det solution and the optimal unitary but not symmetric solution. This gap tends to zero either when the SNR or when the number of BD-RIS elements tends to infinity. Therefore, in these asymptotic regimes, the proposed solution maximizes the achievable rate.

    \item We explore the connections between the low-rank Max-Det solution and its implementation using a $q$-stem topology with a minimum number of reconfigurable impedances. In particular, the Max-Det solution can be optimally implemented with $q$-stem topology with just $q=2r-1$ stems, thus corroborating Corollary 2 in \cite{WuTIT25}. 
 
\end{itemize}

The rest of the paper is structured as follows. In Sec. \ref{sec:Model}, we introduce the scenario of interest and propose maximizing the determinant of the equivalent MIMO channel as a reasonable proxy for rate, showing that it is asymptotically optimal in certain cases. Section \ref{sec:MaxDet} derives a closed-form solution for the Max-Det problem, showing that the optimal rank of the resulting scattering matrix is twice the DoF of the MIMO system. Some log-majorization results are also exploited in Sec. \ref{sec:MaxDet} to upper bound the rate gap between the symmetric BD-RIS that maximizes the determinant and the unitary BD-RIS that maximizes rate (and also the determinant). Sec. \ref{sec:qstem} explores the relationship between the low-rank structure of the scattering matrix and the sparsity of the admittance matrix. Specifically, we show that the low-rank Max-Det solution can be implemented using the recently proposed $q$-stem structure, which requires a minimum number of reconfigurable circuits. Section \ref{sec:simulations} presents simulation results that demonstrate, in practical scenarios, the minor differences between the Max-Det solution and iterative algorithms that maximize the rate, while highlighting a significant reduction in computational cost. Finally, Sec. \ref{sec:conclusions} summarizes the main conclusions and also points out future lines of work.

\indent \textit{Notation}: The symbols for scalars, vectors, matrices, and sets are, respectively, $x$, $\x$, $\X$, and $\mathcal{X}$. $\A^T$, $\A^*$, $\A^H$, $\A^{-1}$, $\A^{1/2}$, $\det(\A)$, $\tr(\A)$, $\rank(\A)$ are, respectively, transpose, conjugate, Hermitian, inverse, square root, determinant, trace and rank of matrix $\A$. $\I_n$ denotes the identity matrix of size $n$ and ${\bf 0}_{a \times b}$ denotes a zero matrix of dimensions $a \times b$. When there is no confusion with the dimensions, we will omit the subscripts. ${\cal CN}({\bf 0}, {\bf R})$ is the proper complex Gaussian distribution with zero mean and covariance matrix ${\bf R}$. $\mathbb{S}_t(s,\mathbb{C}^M)$ denotes the Stiefel manifold of $s$-dimensional frames in a complex space of dimension $M$, i.e.,
 $\mathbb{S}_t(s,\mathbb{C}^M)=\{\Q \in \mathbb{C}^{M\times s}: \Q^H\Q=\I_s\}$. On the other hand, $\CU(n)=\{\U \in \mathbb{C}^{n\times n}: \U\U^H=\I_n\}$ denotes the unitary manifold. We also consider the manifold of unitary and symmetric matrices denoted as  $ \CU_s(n)= \{\U \in \CU(n): \U=\U^T\}$. In our notation, $\log(x)$ denotes the logarithm to the base 2: $\Re(a)$, $\Im(a)$ and $\angle a$ are, respectively, the real part, imaginary part, and the angle of the complex number $a$. Finally, $\sigma_i(\A)$ (sometimes we will use $\sigma_{a_i}$) denotes the $i$-th singular value of the complex matrix $\A$.

\section{System model and problem formulation}
\label{sec:Model}
\subsection{System model}
We consider a BD-RIS-assisted MIMO channel in which the direct channel is blocked. The Tx has $N_t$ antennas and the Rx has $N_r$ antennas. The equivalent channel is
\begin{equation}
\H(\Thetab) = \F \Thetab \G^H
\label{eq:MIMOchanneleq}
\end{equation}
where, $\G \in \mathbb{C}^{N_{t} \times M}$ is the channel from the Tx to the BD-RIS, $\F \in \mathbb{C}^{N_{r} \times M}$ is the channel from the BD-RIS to the Rx, and $\bm{\Theta}$ is the $M \times M $ BD-RIS matrix. In this paper, we denote the degrees of freedom (DoF) of the MIMO channel as $r = \min(N_r,N_t)$, and assume that $M \geq 2 r$ (which is reasonable because $M \gg r $ usually). Although the results of this work are valid for systems with arbitrary $(N_t, N_r)$, we will pay particular attention to the square MIMO case, where $N_t = N_r = r$, due to its simplicity. For a square $r \times r$ MIMO configuration, we introduce the following notation for the Singular Value Decomposition (SVD) of $\F$ and $\G$, which will be useful throughout the paper
\begin{align}
\F &= \underbrace{\begin{bmatrix} \U_{F_1} & \U_{F_2}  \end{bmatrix}}_{\U_F} \begin{bmatrix} \Sigmab_F & {\bf 0}_{r \times (M-r)} \end{bmatrix} \underbrace{\begin{bmatrix} \V_{F_1}^H \vspace{0.1cm} \\ \V_{F_2}^H \end{bmatrix}}_{\V_F^H}  \label{eq:SVDF}\\
\G &= \underbrace{\begin{bmatrix} \U_{G_1} & \U_{G_2}  \end{bmatrix}}_{\U_G}  \begin{bmatrix} \Sigmab_G & {\bf 0}_{r \times (M-r)} \end{bmatrix} \underbrace{\begin{bmatrix} \V_{G_1}^H \vspace{0.1cm} \\ \V_{G_2}^H \end{bmatrix}}_{\V_G^H},  \label{eq:SVDG}
\end{align}
where $\U_{F_1} \in \CU(r)$ and $\U_{G_1} \in \CU(r)$ are $r \times r$ unitary matrices, $\Sigmab_F = \diag(\sigma_{f_1},\ldots, \sigma_{f_r})$ and $\Sigmab_G = \diag(\sigma_{g_1},\ldots, \sigma_{g_r})$ are diagonal matrices with the positive singular values, and, finally, $\V_{F_1} \in \mathbb{S}_t(r,\mathbb{C}^M)$ and $\V_{G_1} \in \mathbb{S}_t(r,\mathbb{C}^M)$ are Stiefel matrices. The dimensions of the rest of matrices are obvious from \eqref{eq:SVDF} and \eqref{eq:SVDG}. Notice that $\F = \U_{F_1} \Sigmab_F  \V_{F_1}^H$ and $\G = \U_{G_1} \Sigmab_G \V_{G_1}^H$ are {\it compact} SVD's for $\F$ and $\G$ (discarding the zero singular values), which will be useful in deriving the solution. In addition, any vector in $\V_{F_2}$ (resp.$\V_{G_2})$ is orthogonal to any vector in $\V_{F_1}$ (resp. $\V_{G_1})$. 

We are interested in passive and symmetric BD-RIS (due to reciprocity) that are characterized by scattering matrices that belong to the feasibility set \cite{ClerckxTWC22b}
\begin{equation}
{{\cal T}} = \{\Thetab \in \mathbb{C}^{M\times M}, \Thetab \Thetab^H \preceq \I_M, \Thetab = \Thetab^T \}.
\label{eq:feasibility}
\end{equation}
In this work, we address the problem of maximizing the rate of the MIMO channel assisted by a BD-RIS in ${{\cal T}}$ and derive the optimal solution at high-SNR. In particular, the optimal $\Thetab$ is a symmetric matrix of rank $2r$ which can be factored as $\Thetab = \Q \Q^T$, with $\Q \in \mathbb{S}_t(2r,\mathbb{C}^M)$. 
Note that the symmetry of $\Thetab$ is a property associated with the reciprocity of the BD-RIS, whereas the low-rank property arises from the solution to an optimization problem. Only the energy incident on a specific subspace spanned by a certain concatenation of $\F$ and $\G$ plays a role in the rate optimization. The fact that the optimal BD-RIS is rank-deficient if $M>2r$ has important implications. The recently proposed $q$-stem topology for BD-RIS implementation \cite{ICC_QStem},\cite{WuTIT25},\cite{Zhou2025Arxiv} has shown that the number of impedances required to interconnect the elements decreases as the rank of $\Thetab$ decreases. A scattering matrix of minimum rank can, therefore, be implemented with the minimum possible number of impedances. For a given $q$-stem implementation, the physically realizable $\Thetab$ it {\it is always full-rank}. Clearly, the full-rank solution scattering matrix can be expressed as the low-rank solution plus its orthogonal complement, which does not affect rate performance. The orthogonal complement of the low-rank solution represents energy radiated outside the subspace of interest to maintain the lossless nature of the BD-RIS. These points will be discussed in greater detail in Sec. \ref{sec:qstem}.
 
\subsection{Rate maximization and its determinant proxy}
\label{sec:detproxy}
The Tx sends proper Gaussian signals with a fixed isotropic covariance matrix, $\x \sim \mathcal{CN}({\bf 0}, \R_{xx})$, with $ \R_{xx} = (P/N_t) \, \I_{N_t}$. The joint optimization of the BD-RIS and the transmit covariance matrix is discussed in Subsection  \ref{sec:Txcovopt}. The received signal is contaminated by additive white Gaussian noise, ${\bf n} \sim \mathcal{CN}({\bf 0}, \sigma^2\, \I_{N_r})$. 

Define the achievable rate as the function\footnote{The expression assumes $r = N_r \leq N_t$, which is a general assumption throughout the paper. The square MIMO case is a special case used to simplify notation in certain proofs, as noted in the text. If $r = N_t < N_r$, $\H(\Thetab)\H(\Thetab)^H $ must be replaced by $\H(\Thetab)^H\H(\Thetab)$.}
\begin{equation}
   f(\Thetab) = \log \det \left( \I_{N_r} + \rho \, \H(\Thetab)\H(\Thetab)^H \right),
\end{equation}
where $\rho = \frac{P}{N_t\sigma^2}$ is the per-antenna SNR, and $\H(\Thetab)$ is given by (1).
Let ${\cal{T}}$ be the feasible set of the BD-RIS scattering matrices in (4), and define the maximum achievable rate as
\begin{equation*}
    f^* = \max_{\Thetab \in {\cal{T}}}  f(\Thetab).
\end{equation*}
We are interested in the solution of the following problem:
\begin{subequations}
\begin{align}
({\cal P}_1):\quad
\argmin_{\Thetab}\quad
& \rank(\Thetab)
\label{eq:rank}
\\
\text{s. t.}\quad
& \Thetab \in {\cal T},
\label{eq:Theta1}
\\
&
f(\Thetab) = f^*.
\label{eq:ropt}
\end{align}
\end{subequations}
In words, we are interested in the BD-RIS in ${\cal{T}}$ of minimum rank that maximizes the achievable rate.

\begin{remark}{\bf [Unitary solution].}
\label{remark:unitary}
We note that the difficulty of $({\cal P}_1)$ stems from the symmetry constraint $\Thetab = \Thetab^T$. Without this constraint, it is well-known that the full-rank unitary solution to the problem aligns the signal subspaces of $\F$ and $\G$ as $\Thetab = \V_F \V_G^H$, which is the solution discussed in \cite{EmilEuCap2025}, while $\Thetab = \V_{F_1}\V_{G_1}^H$ is a minimum rank solution (of rank $s = r = \min(N_t, N_r)$) that achieves the same rate. 
\end{remark}

Let $\H(\Thetab)$ be the $N_r \times N_t$ equivalent channel of rank $r = \min(N_r,N_t)$ and  with singular values $\sigma_{h_1} \geq \ldots \geq \sigma_{h_r} \geq 0$, then the rate can be expressed as \cite{LetaiefTCOM06} 
\begin{multline}    
    \log \det \left( \I_{N_r} + \rho \, \H(\Thetab)\H(\Thetab)^H \right) = r \log(\rho) 
    \\+ \log \det \left(\H(\Thetab)\H(\Thetab)^H \right) + \sum_{i=1}^r \log \left( 1 + \frac{1}{\rho \sigma_{h_i}^2}\right).
    \label{eq:capapp}
\end{multline}
This expression shows that when $\rho \to \infty$ (assuming that the smallest singular value is bounded away from zero, i.e., $\sigma_{h_r} >0$), the last term of \eqref{eq:capapp} vanishes. Consequently, the solution that maximizes the SNR-independent term, $\log \det \left(\H(\Thetab)\H(\Thetab)^H \right)$, converges to the solution that maximizes the rate as the SNR increases. This fact motivates our approach to use determinant maximization as a proxy for rate maximization.  Since $\log(1+x) \leq x/\ln(2)$, the approximation error in \eqref{eq:capapp} can be bounded as
\begin{equation}
    \sum_{i=1}^r \log \left( 1 + \frac{1}{\rho \sigma_{h_i}^2}\right) \leq \frac{1}{\rho \ln 2} \sum_{i=1}^r \frac{1}{\sigma_{h_i}^2} \leq \frac{r}{\rho \, \sigma_{h_r}^2  \ln 2 },
    \label{eq:errorbound}
\end{equation}
suggesting that the error term is small unless $\rho \sigma_{h_r}^2$ is very small. In Sec. \ref{sec:simulations}, we study the impact of using the determinant as a proxy for rate maximization.

In passing, we mention that the use of determinant maximization as a proxy for rate or capacity maximization has been studied in the literature, particularly in the context of antenna subset selection \cite{LetaiefTCOM06}, \cite{GorokhovTSP03}. However, in the context of RIS, it has not been given as much attention. Only \cite{ChoiCL21} considers maximizing the determinant in a line-of-sight (LoS) multi-RIS scenario with diagonal RISs. However, it ultimately exploits the Arithmetic Mean-Geometric Mean (AM-GM) inequality to maximize the Frobenius norm of the equivalent channel instead of its determinant.
\needspace{5\baselineskip}
\begin{remark}{\bf [The max-det proxy is rate optimal when $M \to \infty$].}
\label{remark:Mgrows}
Let us assume\footnote{In RIS-assisted scenarios, $\F$ and/or $\G$ usually have a dominant path, which translates into a deterministic line-of-sight (LoS) component \cite{SemmlerSPAWC}, \cite{SantamariaCL25}, \cite{ZhangRankOneWCL2021}. In this situation, the bulk of the singular values are governed by the random component and concentrate as described in the lemma, while the LoS component acts as a deterministic perturbation that can create several spikes in the distribution of singular values \cite{CouilletBook}, \cite{Couillet11}. Under appropriate normalization, the lemma's argument that the minimum singular value grows with $M$ still holds for Ricean channels.} that the entries of \(\F\) and \(\G\) are i.i.d. complex Gaussian with mean 0 and variance 1. We also assume that \(N_r = N_t = r\) is fixed while \(M \to \infty\), and that \(\Thetab = \Q\Q^T\) is chosen optimally to maximize \(|\det(\F \Thetab \G^H)|\).

Under these assumptions, the singular values \(\sigma_i(\F)\) and \(\sigma_i(\G)\) (for \(i = 1, \dots, r\)) all concentrate around \(\sqrt{M}\) by the Marchenko-Pastur law \cite{Marchenko}, \cite{Tulino} as \(M \to \infty\). More precisely, \(\sigma_{\min}(\F) \sim \sqrt{M} - \sqrt{r} \) and \(\sigma_{\max}(\F) \sim \sqrt{M} + \sqrt{r}\).

Let \(\H(\Thetab) = \F \Thetab \G^H\), with \(\Thetab = \Q\Q^T \preceq \I_M\). Then, by the multiplicative property of singular values \cite{HornBook},
\begin{multline*}
\sigma_{\min}(\H(\Thetab)) \geq \sigma_{\min}(\F) \cdot \sigma_{\min}(\Q \Q^T) \cdot \sigma_{\min}(\G)
\\
\sim (\sqrt{M}) \cdot 1 \cdot (\sqrt{M}) = M \to \infty.    
\end{multline*}
In words, the minimum singular value of the equivalent channel grows unbounded as the number of BD-RIS elements tends to infinity. Therefore, the error term in \eqref{eq:errorbound} vanishes as either $\rho \to \infty$ or $M \to \infty$. In these regimes, the solution that maximizes the determinant also maximizes the rate. Similarly, the minimum singular value of 
\(\H(\Thetab) \H(\Thetab)^H\) grows with $M$ as \(\sigma_{\min}(\H(\Thetab) \H(\Thetab)^H) \sim M^2 \to \infty\).

\end{remark}

Finally, we consider the following Max-Det problem
\begin{subequations}
\begin{align}
({\cal P}_2):\quad
\argmin_{\Thetab}\quad
& \rank(\Thetab)
\label{eq:rank1}
\\
\text{s. t.}\quad
& \Thetab \in {\cal T},
\label{eq:Theta1bis}
\\
&
d(\Thetab) = d^*,
\label{eq:roptd}
\end{align}
\end{subequations}
where $d(\Thetab) = \det \left(\H(\Thetab)\H(\Thetab)^H \right)$ and $d^* = \max\limits_{\Thetab \in {\cal{T}}} d(\Thetab)$. Note that we have eliminated the logarithm, which, being a monotonic function, does not affect the maximization.


\section{Max-Det solution and its rate gap}
\label{sec:MaxDet}
In this section, we present the solution to $({\cal P}_2)$ and analyze its performance.
\subsection{Main result}
\label{sec:MainResult}
\begin{theorem}
\label{th:theoremU}
The optimal solution to $({\cal P}_2)$ is
\begin{eqnarray}
    \Thetab_{opt} &=& \U \begin{bmatrix}\I_r & \0 \\ \0 & -\I_r \end{bmatrix} \U^T \nonumber
    \\
    &=& \underbrace{\U \begin{bmatrix}\I_r & \0 \\ \0 & j\I_r \end{bmatrix}}_{\Q} \underbrace{\begin{bmatrix}\I_r & \0 \\ \0 & j\I_r \end{bmatrix} \U^T}_{\Q^T},
    \label{eq:Thetaopt}
\end{eqnarray}
where $r = \min(N_t,N_r)$ (and hence $\rank(\Thetab_{opt})=2r$ in ${\cal P}_2$) and $\U$ is the left eigenspace obtained from the compact SVD (involving only the positive singular values) of the $M \times 2r$ matrix $\A = \left [\V_{F_1}, \V_{G_1}^* \right ] = \U \Lambdab \V^H$, where ${\bf V}_{F_1}$ are the $r$ main right singular vectors of $\F$, (see \eqref{eq:SVDF}) and ${\bf V}_{G_1}^*$ are the complex conjugate of the $r$ main right singular vectors of $\G$ (see \eqref{eq:SVDG}).
\end{theorem}

\begin{proof}
See Appendix A.
\end{proof}

Theorem \ref{th:theoremU} provides a closed-form solution for the optimal BD-RIS that maximizes $\det \left(\H(\Thetab)\H(\Thetab)^H \right)$, which coincides with the maximum rate solution as the SNR or the number of BD-RIS elements grows. A pseudocode to obtain the optimal Max-Det BD-RIS is shown\footnote{Matlab code can be downloaded from \url{https://github.com/IgnacioSantamaria/Code-MaxDet}.} in Algorithm \ref{alg:snippet}.

\begin{algorithm}[!t]
\small
\DontPrintSemicolon
\SetAlgoVlined
\KwIn{$\F$, $\G$, and $r = \min(N_t,N_r)$}
\KwOut{$\Q$ and $\Thetab = \Q\Q^T$}
Calculate the compact SVDs of $\F$ and $\G$: $\F = \U_{F_1} \Sigmab_F  \V_{F_1}^H$ and $\G = \U_{G_1} \Sigmab_G \V_{G_1}^H$ \;
Form $\A = \left [\V_{F_1} , \V_{G_1}^* \right ]$ \;
Calculate the compact SVD of $\A$: $\A = \U \Lambdab \V^H$\;
Max-Det Solution: $\Q = \U \blkdiag(\I_r, j \I_r)$ and $\Thetab = \Q\Q^T$\;
\caption{Optimal symmetric Max-Det BD-RIS}
\label{alg:snippet}
\end{algorithm}

\begin{remark}{\bf [Genericity of $\F$ and $\G$}]. Theorem \ref{th:theoremU} assumes that $\F$ and $\G$ are generic MIMO channels, i.e., their entries are independently drawn from a continuous distribution. In this situation, the subspaces $\V_{F_1}$ and $\V_{G_1}^*$ do not overlap, and the matrix $\A$ has full column rank with probability 1. 
\end{remark}

\begin{remark}{\bf [Computational complexity].}
Algorithm \ref{alg:snippet} involves two compact SVDs of matrices $\F$ and $\G$, with a computational cost of $\mathcal{O}(M r^2)$ each, and a compact SVD of matrix $\A$ of dimensions $M \times 2r$, with a computational cost of $\mathcal{O}(M (2r)^2)$. Therefore, the complexity is $\mathcal{O}(6 M r^2)$, which scales linearly with the number of RIS elements and quadratically with the number of DoF of the MIMO channel.
\end{remark}

\subsection{Rate penalty of the Max-Det solution}
\label{sec:analysis}
In this subsection, we bound the rate gap between the Max-Det solution and unitary but not necessarily symmetric BD-RIS that maximizes the rate (cf. Remark \ref{remark:unitary}). For a square MIMO channel $\det \left(\H(\Thetab)\H(\Thetab)^H \right) = |\det(\H(\Thetab))|^2$. First, it should be noted that the maximum achievable value for $|\det(\H(\Thetab))|$ is 
\begin{equation}
{\rm D}_{\rm max} = \det(\Sigmab_F)\det(\Sigmab_G) = \prod_{i=1}^r  \sigma_{f_{[i]}} \sigma_{g_{[i]}},
\label{eq:MaxDet}
\end{equation}
where $\Sigmab_F$ and $\Sigmab_G$ contain the non-zero singular values of $\F$ and $\G$ sorted in decreasing order as $\sigma_{f_{[1]}} \geq  ... \geq  \sigma_{f_{[r]}}$ and $\sigma_{g_{[1]}}  \geq ... \geq  \sigma_{g_{[r]}}$, which we shall also denote as $\sigmab_F^{\downarrow}$ and $\sigmab_G^{\downarrow}$. The solution of Algorithm \ref{alg:snippet} generates an equivalent channel whose determinant is ${\rm D}_{\rm max}$, and the rank-$r$ unitary\footnote{Strictly speaking, a rank-$r$ matrix is not unitary. We use this term loosely, understanding it to be a matrix with orthogonal columns and with $r$ singular values equal to 1 and the remaining singular values equal to 0. Remember also that it is possible to complete the subspaces $\V_{F_1}$ and $\V_{G_1}$ with their orthogonal complements to form a unitary matrix without modifying the solution.}, but not symmetric, solution $\Thetab = \V_{F_1}\V_{G_1}^H$ also attains ${\rm D}_{\rm max}$. The last point can be easily verified by substituting the compact SVD's of $\F$ and $\G$ into the MIMO equivalent channel, i.e.,
\begin{align}
\nonumber    
    \H(\Thetab) &= \F \Thetab \G^H = \U_F \Sigmab_F \V_{F_1}^H \underbrace{\V_{F_1}\V_{G_1}^H}_{\Thetab} \V_{G_1} \Sigmab_G \U_G^H
    \\ &= \U_F \Sigmab_F  \Sigmab_G \U_G^H. \label{eq:Htheta}
\end{align}
Although we assume a square MIMO system for simplicity, the results remain valid when $N_r \neq N_t$. If, for example, $N_r>N_t=r$, it suffices to select the first $r$ columns of $\U_{F_1}$  and the $r \times r$ submatrix of $\Sigmab_F$. Now, the absolute value of the determinant of $\H(\Thetab)$ is
\begin{align*}
    | \det(\H(\Thetab)) | &= | \det(\U_F) \det(\U_G^H) | \det(\Sigmab_F \Sigmab_G) 
    \\
    &= \prod_{i=1}^r  \sigma_{f_{[i]}} \sigma_{g_{[i]}} = {\rm D}_{\rm max},
\end{align*}
where, in the last equality, we have used the property that if $\U \in \CU(r)$ is a unitary matrix, then $\U^H\U = \I_r$, implying that $|\lambda_i| =1$, $i=1,\ldots,r$; where $\lambda_i$ denotes the $i$-th eigenvalue of $\U$. Therefore, $|\det(\U_F)| = |\det(\U_G^H)| =1$.
 
Certainly, many other possible solutions reach ${\rm D}_{\rm max}$ even though the singular values of the corresponding equivalent channels may be different. It is possible to arbitrarily rotate the subspaces $\V_{F_1}$ and $\V_{G_1}$ to build a new rank-$r$ unitary matrix, $\Thetab' = \V_{F_1} \U \V_{G_1}^H$ with $\U \in \CU(r)$, such that $|\det(\H(\Thetab'))| =  |\det(\H(\Thetab))| = {\rm D}_{\rm max}$. This can be easily verified by repeating the steps in \eqref{eq:Htheta}, but now introducing another unitary matrix, $\U$, sandwiched between $\V_{F_1}$ and $\V_{G_1}$. Since $|\det(\U)|=1$, the determinant of the equivalent channel matrix does not change, but its singular values do. By varying $\U \in \CU(r)$ in $\Thetab' = \V_{F_1} \U \V_{G_1}^H$, it is possible to explore all equivalent channels with the same maximum determinant. 

By relaxing the BD-RIS symmetry constraint, it is possible to prove a somewhat stronger result. Proposition 3 in \cite{Zhao25} shows that, under the assumption of a negligible direct channel, there is a unitary but nonsymmetric BD-RIS shaping an equivalent channel whose singular values may be exactly the same as those of $\Sigmab_F \Sigmab_G$. The symmetry constraint significantly limits this result. For a symmetric and unitary BD-RIS, it is not possible to shape the equivalent channel to get the singular values of $\Sigmab_F \Sigmab_G$, but it is possible to do so for its determinant.

Looking at the problem from this perspective, we are faced with the following question: of all the equivalent channels generated with BD-RISs satisfying $\Thetab \Thetab^H  \preceq \I_{M}$, which maximizes rate and why? As we will see, the answer to this question allows us to bound the rate gap between the optimal symmetric and unitary solutions. Before presenting the results, we define the following concepts. 

\begin{definition}[{\bf Log-majorization} \cite{Palomar2011}]
    For two vectors \(\x, \y \in \mathbb{R}^r_+\) with components sorted in decreasing order $\x^{\downarrow} = (x_{[1]} \geq \ldots \geq x_{[r]})$ and similarly for \(\y\), \(\x\) is log-majorized by \(\y\) (denoted \(\x \prec_{\log} \y\)) if: 
\[
\prod_{i=1}^k x_{[i]} \leq \prod_{i=1}^k y_{[i]} \quad \text{for } 1 \leq k < r,
\]
and
\[
\prod_{i=1}^r x_{[i]} = \prod_{i=1}^r y_{[i]}.
\]
This is equivalent to the standard (additive) majorization (see \cite{JorswieckMonograph}, \cite{OlkinBook}) of the component-wise logarithms: \(\log \x \prec \log \y\), where \(\log\) is applied element-wise. 
\end{definition}

\begin{definition}[{\bf Schur multiplicatively convex functions} \cite{CHU2012412}]
Let $E \subseteq \mathbb{R}^r_+$ be a set. A real-valued function $f: E \mapsto \mathbb{R}_+$ is said to be Schur multiplicatively convex on $E$ if $f(\x) \leq f(\y)$ for each pair of $r$-tuples in $E$ such that $\x \prec_{\log} \y$. $f$ is said to be Schur multiplicatively concave if $1/f$ is Schur multiplicatively convex, and, hence, $f(\x) \geq f(\y)$. 
\end{definition}

We now present the following result:
\begin{lemma}
\label{lemma:logmaj}
Let $\H(\Thetab) = \F\Thetab \G^H$ and $\H(\Thetab') = \F \Thetab' \G^H$ be two MIMO channels generated with two different BD-RIS such that $|\det(\H(\Thetab))| = |\det(\H(\Thetab'))|$. The singular values in decreasing order of $\H(\Thetab)$ and $\H(\Thetab')$ are $\sigmab^{\downarrow}$ and $\sigmab'^{\downarrow}$, respectively. Then, if $\sigmab'^{\downarrow} \prec_{\log} \sigmab^{\downarrow}$, 
\[\log\det(\I_r + \rho \H(\Thetab) \H(\Thetab)^H) \geq \log\det(\I_r + \rho \H(\Thetab') \H(\Thetab')^H).
\]
\end{lemma}

\begin{proof} The function
    \[f(\sigmab)=\sum_{i=1}^r \log(1+\rho \,\sigmab_i^2) \qquad (\sigmab_i>0)\]
is Schur-multiplicative convex in $\sigmab$ for positive $\rho$.
This follows from the log-transform
\[H(\u)=f(e^{u_1},\dots,e^{u_r})=\sum_{i=1}^r \phi(u_i), 
\]
with $\phi(u)= \, \log \, \bigl(1+\rho\,e^{2u} \bigr)$, where the function $\phi$ is convex for $\rho>0$. Additionally, it can be checked that the function satisfies 
\begin{equation*}
    (\log \sigma_1 - \log \sigma_2) \left( \sigma_1 \frac{\partial f}{ \partial\sigma_1} -  \sigma_2 \frac{\partial f}{ \partial\sigma_2} \right) \geq 0,
\end{equation*}
which is the condition for $f$ to be Schur multiplicatively convex (see Lemma 2.2 in \cite{CHU2012412}).
\end{proof}

Lemma \ref{lemma:logmaj} shows that, of all equivalent channels with the same determinant, the optimal one in terms of rate is the one whose singular values log-majorize those of any other. This optimal solution is $\sigmab_H^{\downarrow} = \sigmab_F^{\downarrow} \sigmab_G^{\downarrow}$, obtained by the BD-RIS $\Thetab = \V_{F_1}\V_{G_1}^H$, which achieves an optimal matching between the eigenmodes of $\F$ and those of $\G$. Therefore, among all rank-$r$ matrices satisfying $\Thetab^H\Thetab \preceq \I_M$, $\Thetab = \V_{F_1}\V_{G_1}^H$ maximizes the rate and also maximizes the determinant (and, incidentally, also the Frobenius norm) of the equivalent channel. The optimal symmetric BD-RIS in Theorem \ref{th:theoremU} gives an equivalent channel whose singular values are log-majorized by those of the optimal unitary solution with optimal eigenmode matching. The following result bounds the rate penalty between the optimal symmetric and unitary solutions, showing that it disappears when the SNR or the number of elements of the BD-RIS grows unbounded.

\begin{lemma}{\bf [Rate gap].}
\label{lemma:capgap}
The rate gap between the optimal unitary BD-RIS and the optimal symmetric BD-RIS can be bounded as
\begin{equation}
    \Delta R \leq r \log \left( \frac{(1 + \rho \sigma_{f_r}^2 \sigma_{g_r}^2) \sigma_{f_1}^2 \sigma_{g_1}^2}{(1 + \rho \sigma_{f_1}^2 \sigma_{g_1}^2)  \sigma_{f_r}^2 \sigma_{g_r}^2} \right),
    \label{eq:DeltaC}
\end{equation}
where $\rho = P/(N_t\sigma^2)$ is the per-antenna SNR, and $\sigma_{f_1}(\sigma_{f_r}$) and $\sigma_{g_1}(\sigma_{g_r})$ are the max (min) singular values of $\F$ and $\G$, respectively

Furthermore, $ \Delta R \to 0$ as $\rho \to \infty$ or $M \to \infty$.
\end{lemma}

\begin{proof}
A corollary that follows from Lemma \ref{lemma:logmaj} is that, among all equivalent channels with maximum determinant, the worst-case in terms of rate is the channel whose singular values are log-majorized by any other solution with equal determinant. Let  $\H(\Thetab) = \F \Thetab \G^H$ be an equivalent channel generated by $\Thetab$ with singular values $\sigmab^{\downarrow}$, such that $|\det(\H(\Thetab))| = \prod_{i=1}^r \sigma_{f_i} \sigma_{g_i}$. Then, $\sigmab^{\downarrow}$ satisfies the inequality \cite[Chapter 9, Theorem H.1]{OlkinBook}
\begin{equation}
\sigmab^{\downarrow} \prec_{\log}  \sigmab_F^{\downarrow}  \sigmab_G^{\downarrow}.
\label{eq:logbound}
\end{equation}

Now, we can use the rate expression in \eqref{eq:capapp} to bound the rate gap as
\begin{align*}
   \Delta R &=   \sum_{i=1}^r \log \left( 1 + \frac{1}{\rho (\sigma_{f_i} \sigma_{g_i})^2}\right) -  \sum_{i=1}^r \log \left( 1 + \frac{1}{\rho \sigma_i^2}\right) \\
   & \overset{(a)}{\leq } r \left(\log \left( 1 + \frac{1}{\rho (\sigma_{f_r} \sigma_{g_r})^2}\right) - \log \left( 1 + \frac{1}{\rho \sigma_{\rm max}^2}\right) \right) \\
    & \overset{(b)}{\leq } r \left(\log \left( 1 + \frac{1}{\rho (\sigma_{f_r} \sigma_{g_r})^2}\right) - \log \left( 1 + \frac{1}{\rho (\sigma_{f_1} \sigma_{g_1})^2}\right) \right) 
    \label{eq:bounds}
\end{align*}
where in $(a)$ we have taken the largest term of the first summation (corresponding to the smallest $\sigma_{f_r} \sigma_{g_r}$) and the smallest term of the last summation (corresponding to the largest $\sigma_{\rm max}$ of the worst-case channel, while in $(b)$ we have used the inequality in \eqref{eq:logbound} to bound $\sigma_{\rm max} \leq \sigma_{f_1} \sigma_{g_1}$. From this result, the inequality \eqref{eq:DeltaC} follows. It is easy to check that $\Delta R$ goes to zero when $\rho \to \infty$ or $M \to \infty$ (since all singular values concentrate around $\sqrt{M}$).
\end{proof}

\begin{remark}{\bf [Bound tightness]}
  The bound $\sigma_{\rm max} \leq \sigma_{f_1} \sigma_{g_1}$ is not necessarily tight for the worst channel in terms of rate. To derive tighter bounds, a complete characterization of the singular values of $\H(\Thetab) = \F \Thetab \G^H$ in terms of those of $\F$ and $\G$ is needed, which is given by the so-called multiplicative Horn inequalities \cite{BERCOVICI2015400},\cite{KLYACHKO200037},\cite{Loyka15}, \cite{Zhao25}. In particular, a tighter bound can be derived from the inequality $\sigma_{\rm max} \leq \max (\sigmab_F^{\downarrow}  \sigmab_G^{\uparrow})$, in which the singular values in descending order of $\F$ are aligned with the singular values in ascending order of $\G$.  
\end{remark}

Let us illustrate the results of this subsection with a simple example. We consider a $2 \times 2$ MIMO system ($N_t=N_r=r=2$) assisted by a BD-RIS with $M=16$ elements. The entries of the forward, $\G$, and backward, $\F$ channels are i.i.d. complex Gaussian with mean 0 and variance 1. The curve in Fig. \ref{fig:SVs1} represents the singular values of the equivalent channels generated as
\begin{equation*}
    \H(\Thetab') = \F \left( \underbrace{\V_{F_1} \begin{bmatrix} \cos(\phi) & -\sin(\phi) \\ \sin(\phi) & \cos(\phi)
    \end{bmatrix} \V_{G_1}^H}_{\Thetab'} \right) \G^H, 
\end{equation*}
where $0\leq \phi \leq \pi/2$.
The equivalent channel obtained for $\phi= 0$ optimally aligns the eigenmodes of $\F$ with those of $\G$, maximizing both the determinant of the equivalent channel and the rate. The singular values of this channel correspond to the point furthest to the right on the curve in Fig. \ref{fig:SVs1}. When $\phi$ varies from 0 to $\pi/2$, the singular values of the equivalent channel travel along the curve from right to left, obtaining other alignments of the eigenmodes, which also maximize the determinant but not the rate. The symmetric solution described in Theorem \ref{th:theoremU} yields the intermediate point shown in blue in Fig. \ref{fig:SVs1}.
 
\begin{figure}[htb]
    \centering
\includegraphics[width=.47\textwidth]{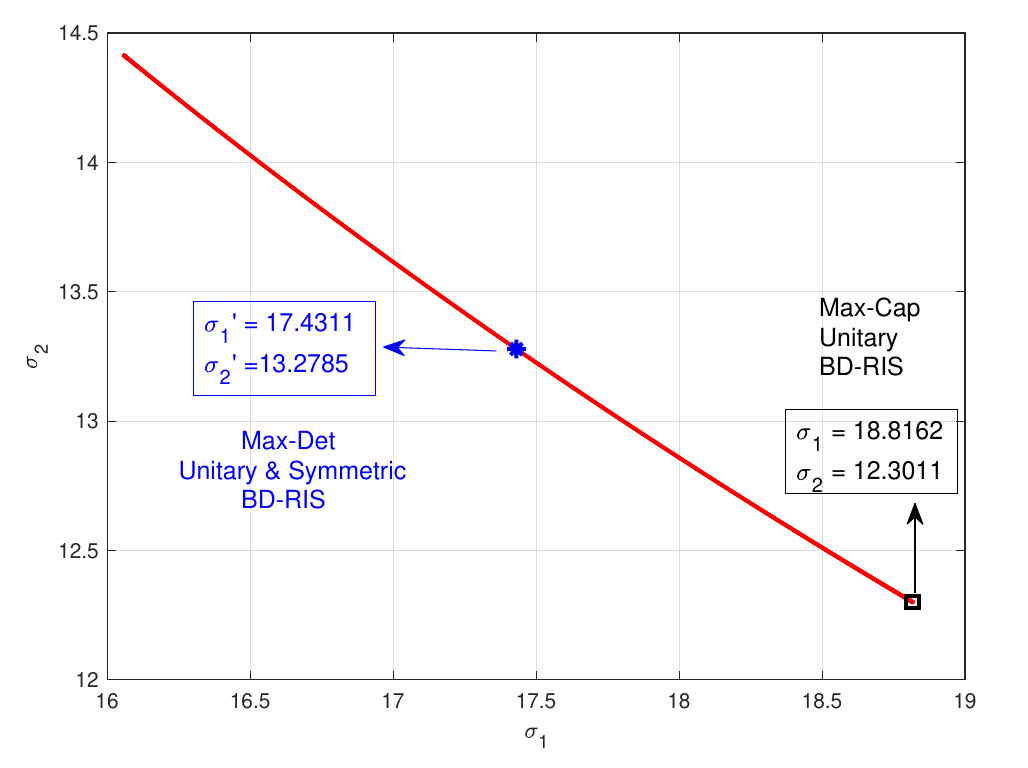}
     \caption{The curve represents the singular values of all equivalent channels $\H(\Thetab) = \F \Thetab \G^H$ with maximum determinant for a $2 \times 2$ MIMO system assisted by a BD-RIS with $M=16$ elements.}
	\label{fig:SVs1}
\end{figure}

For this example, the rate gap between the unitary solution and the symmetric solution at $\rho = 1$ (0 dB) and $\rho = 100$ (20 dB) is $\Delta R = 6.72 \cdot 10^{-4}$ and $\Delta R = 6.78 \cdot 10^{-6}$, respectively, while the bound provided in Lemma \ref{lemma:capgap} yields $1.9  \cdot 10^{-2}$ (for $\rho = 1$) and $1.9  \cdot 10^{-4}$ (for $\rho = 100$). This suggests that the penalty in rate is generally negligible, which will be confirmed in Sec. \ref{sec:simulations}.
\subsection{Unblocked direct channel} The solution in Theorem \ref{th:theoremU} is only optimal if the direct channel is blocked. When a direct channel between Tx and Rx exists, the equivalent channel is $\H_d + \F \Thetab \G^H$, and a closed-form solution for the BD-RIS that maximizes the determinant (let alone the rate) is unknown and seems unlikely.In the following paragraphs, we will first conduct a sensitivity analysis of the Max-Det solution to a residual direct link and, second, propose a suboptimal solution to the maximum rate problem with an unblocked direct link, which, as will be seen in Sec. \ref{sec:simulations}, offers a competitive solution.

\paragraph{Sensitivity analysis} Consider the perturbed channel ${\bf P}(\Theta) = \H_d + \H(\Theta)$, where $\H_d$ is a residual direct-link and $\H(\Thetab) = \F \Thetab \G^H$ is the BD-RIS channel. Let $\Thetab_{opt}$ be the Max-Det solution with no direct link and define $\A = \H(\Thetab_{opt})\H(\Thetab_{opt})^H$ and $\A_{pert} = {\bf P}(\Thetab_{opt}){\bf P}(\Thetab_{opt})^H$. The question we ask is the following: if $\|\H_d\|_F$ is small, how much can $\det(\A_{pert})$ deviate from $\det(\A)$? 
The deviation can be bounded by applying Weyl’s perturbation inequality for the singular values of two arbitrary matrices of the same size \cite[Corollary~7.3.5]{HornBook}, which states that the ordered singular values cannot change by more than the spectral norm of the additive perturbation 
\begin{equation}
\big|\sigma_i({\bf P}(\Thetab_{opt})) - \sigma_i)\big| \leq \|\H_d\|_2 \leq \|\H_d\|_F, \quad \forall i,
\label{eq:Weyl}
\end{equation}
where $\sigma_i$ are the singular values of $\H(\Thetab_{opt})$.
Let us assume that $\|\H_d\|_F \le \epsilon$. Then, squaring \eqref{eq:Weyl} and taking the product over $i$ yields
\begin{equation*}
\prod_{i=1}^{N_r} (\sigma_i - \epsilon)_+^2 \;\le\; \det(\A_{pert}) \;\le\; \prod_{i=1}^{N_r} (\sigma_i + \epsilon)^2 ,
\end{equation*}
where $(x)_+ = \max(x,0)$. A more informative bound is obtained by assuming that $\epsilon < \sigma_{N_r}$ and bounding every factor by the smallest singular value
\begin{equation*}
\Big(1-\frac{\epsilon}{\sigma_{N_r}}\Big)^{2N_r} \;\le\; \frac{\det(\A_{pert})}{\det(\A)} \;\le\; \Big(1+\frac{\epsilon}{\sigma_{N_r}}\Big)^{2N_r}.
\end{equation*}
We conclude that the sensitivity to a residual direct link is therefore larger precisely when the BD-RIS-only channel is weak or ill-conditioned in some direction. A well-conditioned, strong BD-RIS design is inherently more robust to a (small) direct-link channel.

\paragraph{A suboptimal solution}
Note that the solution given in Theorem \ref{th:theoremU}, $ \Thetab_{opt}$, can be rotated arbitrarily as $e^{j \phi}\Thetab_{opt}$ without changing $|\det(\F\Thetab_{opt}\G^H)|$. This arbitrary phase, $\phi$, can be chosen as the solution to
\begin{multline}
    ({\cal P}_3): \,\max_{\phi \in [0, 2\pi)]}\,\, \log \det \big( \I_{N_r} + \rho \, (\H_d + e^{j \phi} \H(\Thetab_{opt})  )\\
    \times (\H_d + e^{j \phi}  \H(\Thetab_{opt}) )^H \big), \label{eq:phi}
\end{multline}
where $\H(\Thetab_{opt}) = \F \Thetab_{opt} \G^H$. 

\subsection{Optimization of the Tx covariance matrix}
\label{sec:Txcovopt}
The solution to ${\cal{P}}_2$ assumes a Gaussian isotropic transmitted signal. In this subsection, we consider the optimization of the transmission covariance matrix under the criterion of maximizing the determinant. We assume, as in the rest of the paper, that $N_r \leq N_t$.
Consider a fixed full row rank MIMO channel, $\H \in \mathbb{C}^{N_r \times N_t}$ and the problem
\begin{equation*}
\max_{\R_{xx} \succeq 0,\ \tr(\R_{xx}) \le P} \ \det \!\big(\H \R_{xx} \H^H\big).
\end{equation*}
This is different from the MIMO capacity problem, whose solution is the classical water-filling covariance, which depends on the noise variance $\sigma^2$.

Let $\H = \U \Sigmab \V^H$ be the SVD, with $\Sigmab = \begin{bmatrix} \D & 0 \end{bmatrix}$ and $\D = \diag(\sigma_1,\dots,\sigma_{N_r})$, $\sigma_i > 0$. Define
\begin{equation*}
\tilde{\R}_{xx} = \V^H \R_{xx} \V,
\end{equation*}
which is positive semidefinite and has the same trace as $\R_{xx}$ since $\V$ is unitary. Partition
\begin{equation*}
\tilde{\R}_{xx} = \begin{bmatrix}\tilde{\R}_{11} & \tilde{\R}_{12} \\ \tilde{\R}_{12}^H & \tilde{\R}_{22} \end{bmatrix},
\end{equation*}
with $\tilde{\R}_{11} \in \mathbb{C}^{N_r \times N_r}$. Then
\begin{equation*}
\H \R_{xx} \H^H = \U \D \tilde{\R}_{11} \D^H \U^H,
\end{equation*}
and
\begin{equation*}
\det(\H\R_{xx}\H^H) = \Big(\prod_{i=1}^{N_r} \sigma_i^2\Big)\, \det(\tilde{\R}_{11}).
\end{equation*}
Therefore, to maximize the determinant all available power should be placed in $\tilde{\R}_{11}$. In addition, by the arithmetic-geometric inequality, the optimal solution is $\tilde{\R}_{11}^\star = \frac{P}{N_r} \I_{N_r}$. Writing $\V = [\,\V_1 \ \V_2\,]$ where $\V_1$ collects the first $N_r$ columns of $\V$ (associated with the nonzero singular values of the MIMO channel), the Tx covariance matrix maximizing the determinant is
\begin{equation*}
\tilde{\R}_{xx}^\star = \frac{P}{N_r}\, \V_1 \V_1^H.
\end{equation*}
If $\H$ is square and full rank, this reduces to $\tilde{\R}_{xx}^\star = \frac{P}{N_t} \I_{N_t}$, which was the assumption made in Sec. \ref{sec:detproxy}.

The result is not surprising, since we know that the determinant is a good proxy for the rate at high-SNR, and in that regime, uniform allocation of power among the channel modes is optimal. Perhaps it is more interesting to note that when the MIMO channel depends on the BD-RIS, the joint problem
\begin{equation*}
\begin{aligned}
&\max_{\Thetab}\; \max_{\R_{xx}}\;
\det\!\bigl(\H(\Thetab)\R_{xx}\H(\Thetab)^H\bigr) \\
&\qquad =
\left(\frac{P}{N_r}\right)^{N_r}
\max_{\Thetab}\;
\det\!\bigl(\H(\Thetab)\H(\Thetab)^H\bigr).
\end{aligned}
\end{equation*}
is separable: the inner optimum over $\R_{xx}$ is a constant multiple of $\det(\H(\Thetab)\H(\Thetab)^H)$, independent of which $\H(\Thetab)$ results from $\Thetab$. In conclusion, an alternating optimization procedure is not necessary to maximize the determinant.
However, the goal is to maximize the rate, not the determinant. Therefore, if the equivalent channel is known at the transmitter, we propose the following two-step (non-iterative) solution:
\begin{enumerate}
\item Compute $\Thetab_{opt}$ as the closed-form solution that maximizes the determinant.
\item Compute $\R_{xx}^\star$  by transmitting along the right singular vectors of $\H(\Thetab_{opt})$ and allocating power across those eigenmodes by water-filling.
\end{enumerate}
As we will show in Sec. \ref{sec:simulations}, the loss of this procedure compared to a joint alternating optimization method is significant only at very low SNRs.

\section{Max-Det BD-RIS implementation through a $q$-stem architecture}
\label{sec:qstem}
A BD-RIS is implemented through an $M$-port network interconnected through a reconfigurable admittance matrix, $\Y$, which is the inverse of the impedance \cite{ClerckxTWC22a}, \cite{Nossek24}. To maximize the power reflected by the intelligent surface, the admittance matrix must be purely imaginary, i.e., $\Y = j\B$, where $\B$ is the so-called susceptance matrix \cite{Nerini24}. \cite{Nerini2026surveyBDRIS}. The scattering matrix, $\Thetab$, and the susceptance matrix, $\B$, are related by
\begin{equation}
    \Thetab = (\I_M + jZ_0 \B)^{-1}(\I_M - jZ_0 \B),
    \label{eq:B2Theta}
\end{equation}
where $Z_0$ is the reference impedance, usually set to  $Z_0 =  50 \, \Omega$. Note that the elements $b_{ii}$, $i=1,\ldots, M$ of $\B$ represent the impedances connecting each reflecting element to ground, while the elements $b_{ij}$ ($i \neq j$) model the impedances connecting elements $i$ and $j$. In the {\it fully-connected} architecture \cite{ClerckxTWC22a}, \cite{ClerckxTWC22b}, \cite{ClerkxBDRISTut26}, all reflecting elements are interconnected and, consequently, $\B$ is a real and symmetric $ M \times M$ matrix with a total number of $M(M+1)/2$ reconfigurable circuits. More recently, the $q$-stem structure, proposed in \cite{ICC_QStem} and analyzed in detail in \cite{WuTIT25}, \cite{Zhou2025Arxiv}, \cite{Nerini24}, has shown to be capable of achieving the optimal performance of fully connected networks in multi-user MIMO
systems with a significantly lower number of impedances. The susceptance matrix of the $q$-stem topology is structured as
\begin{equation}
        \B = \begin{bmatrix} \B_{11} & \B_{12}^T \\ \B_{12}  & \Gammab \end{bmatrix},
        \label{eq:qstem}
\end{equation}
where $\B_{11} = \B_{11}^T$ is a $q \times q$  symmetric and real matrix, $\B_{12}$ is a $(M-q) \times q$ block and $\Gammab = \diag(b_{qq+1},\ldots b_{M}))$ is a $(M-q) \times (M-q)$ diagonal matrix. Therefore, the lower-right off-diagonals of $\B$ in the $q$-stem topology are zero, thus reducing the total number of reconfigurable elements to $\nu = q(q+1)/2 + (M-q)(q+1)$. Fig. \ref{fig:qstemNew} shows an example of the structure of the susceptance matrix, $\B$, with $M = 6$ elements and $q = 2$ stems.

Theorem 1 and Corollary 2 in \cite{WuTIT25} prove that the $2r$ low-rank Max-Det solution can be implemented by a $q$-stem topology with $q = 2r-1$ stems. The following lemma presents an alternative, yet simple, proof of the same result connecting the sparsity of the $q$-stem topology with the low-rank structure in the scattering matrix.

\begin{figure}[ht]
    \centering

\begin{tikzpicture}[
  mycell/.style={
    draw,
    minimum width=7mm,
    minimum height=7mm,
    inner sep=0pt,
    anchor=center
  }
]
\matrix (M) [matrix of math nodes,
             nodes=mycell,
             row sep=-\pgflinewidth,
             column sep=-\pgflinewidth] {
  |[fill=gray!80]| b_{11} & |[fill=blue!25]| b_{21} & |[fill=red!25]| b_{31} & |[fill=green!25]| b_{41} & |[fill=yellow!30]| b_{51} & |[fill=gray!40]| b_{61} \\
  |[fill=blue!25]| b_{21}  & |[fill=red!55]|  b_{22}  & |[fill=green!10]| b_{32} & |[fill=magenta!25]| b_{42} & |[fill=orange!20]| b_{52} & |[fill=purple!50]| b_{62} \\
  |[fill=red!25]| b_{31}  & |[fill=green!10]| b_{32} & |[fill=yellow!65]| b_{33} & |[fill=none]| 0 & |[fill=none]| 0 & |[fill=none]| 0 \\
  |[fill=green!25]| b_{41} & |[fill=magenta!25]| b_{42} & |[fill=none]| 0 & |[fill=blue!50]| b_{44} & |[fill=none]| 0 & |[fill=none]| 0 \\
  |[fill=yellow!30]| b_{51} & |[fill=orange!20]| b_{52} & |[fill=none]| 0 & |[fill=none]| 0 & |[fill=green!50]| b_{55} & |[fill=none]| 0 \\
  |[fill=gray!40]| b_{61} & |[fill=purple!50]| b_{62} & |[fill=none]| 0 & |[fill=none]| 0 & |[fill=none]| 0 & |[fill=cyan!10]| b_{66} \\
};
\end{tikzpicture}

 \caption{Structure of a $q$-stem matrix $\B$ with $M=6$ and $q=2$.}
    \label{fig:qstemNew}
\end{figure}

\begin{lemma}{\bf [$q$-stem structure for the Max-Det solution].}
\label{lemma:qstem}
Let $\B$ be a real and symmetric $M \times M$ matrix with the $q$-stem structure as in \eqref{eq:qstem}, and let $\Thetab_{lr} = \Q \Q^T$ be the rank-$2r$ Max-Det solution. If $q=2r-1$, there exists a $\B$ with corresponding scattering matrix $\Thetab_{\B} = (\I_M + jZ_0 \B)^{-1}(\I_M - jZ_0 \B)$ such that $|\det(\F\Thetab_{\B}\G^H)| = |\det(\F\Thetab_{lr}\G^H)|$. Therefore, $\Thetab_{\B}$ is a Max-Det solution that can be implemented with the $q$-stem topology.
\end{lemma}

\begin{proof}
For $\Thetab_{\B}$, generated with a $q$-stem $\B$, to have identical performance to the Max-Det low-rank $\Thetab_{lr} = \Q \Q^T$, it must be possible to decompose it as $\Thetab_{\B} = \Q \Q^T + \Q_{\perp} \Q_{\perp}^T$, where $\Q_{\perp}$ is the orthogonal complement of $\Q$. Post-multiplying by $\Q^*$, this yields the condition $\Thetab_{\B}\Q^* = \Q$. Substituting $\Thetab_{\B} = (\I_M + jZ_0 \B)^{-1}(\I_M - jZ_0 \B)$ in this expression we get
\begin{equation}
 (\I_M - jZ_0 \B)\Q^* =  (\I_M + jZ_0 \B)\Q, 
\end{equation}
which can be rewriten as $\B \C = - \D$, where $\C = jZ_0(\Q^* + \Q) = j2Z_0 \Re(\Q) $ and $\D = jZ_0( \Q - \Q^*) = j2Z_0 \Im(\Q)$, which is\footnote{However, note that there is a minus sign in our equation that does not appear in \cite[Eq. (32)]{Zhou2025Arxiv}.} essentially Eq. (32) in \cite{Zhou2025Arxiv}. Therefore, the $q$-stem matrix $\B$ is the solution to the linear system of equations $\B \Re(\Q) = -\Im(\Q)$. Using standard vectorization operations, the system can be rewritten as
\begin{equation}
    \underbrace{\left(\Re(\Q)^T \otimes \I_M \right) {\bf R}}_{\W} {\bf b} = - \vec(\Im(\Q)),
    \label{eq:system}
\end{equation}
where ${\bf R}\b = \vec(\B)$, with ${\bf R}$ an $M^2 \times \nu$ selection matrix that places the zeros of $\B$ in the right position (cf. \cite{Zhou2025Arxiv}) and $\b$ is a vector with the $\nu = q(q+1)/2 + (M-q)(q+1)$ nonzero elements of $\B$. The linear system in \eqref{eq:system} has $2Mr$ equations and $\nu$ unknowns. However, not all equations are linearly independent. The reason is that $\Re(\Q)$ and $\Im(\Q)$ are constrained in such a way that $\Q =\Re(\Q) + j \Im(\Q)$ belongs to $\mathbb{S}_t(2r,\mathbb{C}^M)$, the complex Stiefel manifold of $2r$-dimensional frames in $\mathbb{C}^M$. The number of real dimensions of $\mathbb{S}_t(2r,\mathbb{C}^M) $ is ${\rm dim} (\mathbb{S}_t(2r,\mathbb{C}^M)) = 4rM -4r^2$ \cite[Ch. 9]{Coherence}. Therefore, the number of independent variables in $\Im(\Q)$ is $2rM -2r^2$. All in all, \eqref{eq:system} is a linear system of $2rM -2r^2$ independent equations with $\nu = q(q+1)/2 + (M-q)(q+1)$ unknowns. Now, it is a trivial exercise to show that when $q= 2r-1$, the number of unknowns is $\nu = 2rM -2r^2 + r$, which is larger than the number of equations, thus ensuring a solution exists because $\W$ has full column rank generically \footnote{In fact, this simple analysis suggests that it would be possible to make zero $r$ additional values of the $q$-stem structure, and still solve \eqref{eq:system} exactly. This would be the Max-Det implementation with the minimum number of reconfigurable circuits: $2rM -2r^2$.}. This completes the proof. 
\end{proof}

\section{Simulation Results}
\label{sec:simulations}
In this section, we show the performance of the proposed solution, comparing it with other baselines in different scenarios. 
\subsection{Scenario description}
\label{sec:scenario}
We consider a $4 \times 4$ MIMO system with the Tx located at coordinates (0,0,1.5) [m] and the Rx at coordinates (50,0,1.5) [m]. A BD-RIS with $M$ elements is located close to the Tx at coordinates (5,3,3) [m] to assist the Tx-Rx communication. The channels through the BD-RIS have a dominant LoS path and are therefore modeled as Rician with factor $K=2$ and path loss exponent $\alpha = 2$. We consider scenarios with and without a direct channel. When the direct channel is not obstructed, it is modeled as a Rayleigh channel with path loss exponent $\alpha = 4$. The transmit covariance matrix is isotropic with uniform power allocation across eigenmodes $\R_{xx} = \frac{P}{N_t} \I_{N_t}$. As a reference, we define the SNR\footnote{Note that the SNR at the Rx depends on the equivalent MIMO channel and, consequently, on the BD-RIS. As a reference, we have taken the SNR for a BD-RIS $\Thetab = \I_M$. However, our SNR definition is completely immaterial to the results.} 
\begin{equation}
    {\rm SNR} = 10 \log_{10} \frac{P \| \F \G^H\|_F^2}{N_t N_r \sigma^2}. 
\end{equation}
The results shown are the average of 200 independent simulations (channel realizations) keeping the Tx, Rx and BD-RIS positions fixed. 

\subsection{Maximizing the determinant approaches the achievable rate when $M$ grows}

It is well known that maximizing the determinant is optimal in terms of rate as the SNR increases. The fact that it is also optimal as the number of BD-RIS elements increases is more relevant to this work. In this first simulation, we analyze the percentage of the achievable rate attained by the Max-Det solution as $M$ increases, for different SNRs. To reproduce the conditions of Remark \ref{remark:Mgrows}, we have generated Rayleigh channels, $\F$ and $\G$, with i.i.d. ${\cal {CN}}(0,1)$ entries. The SNR is determined by the parameter $\rho$ as in \eqref{eq:capapp}. To maximize the achievable rate, we use the optimization algorithm on the manifold of unitary and symmetric matrices proposed in \cite{santamariaICASSP26} and generalized in \cite{SantamariaUsArxiv}. After optimizing the BD-RIS, we obtain the optimal transmission covariance matrix through waterfilling on the equivalent channel. Fig. \ref{fig:RatePercentagevsM_6x3} shows the percentage of the achievable rate attained by the Max-Det solution for increasing values of $M$. Even for $M=16$ and $\rho= 10^{-4}$, the Max-Det solution already reaches 84\% of the achievable rate. The percentage quickly converges to 100\% as $M$ increases, converging more rapidly as the SNR increases.

\begin{figure}[htb]
    \centering
\includegraphics[width=.47\textwidth]{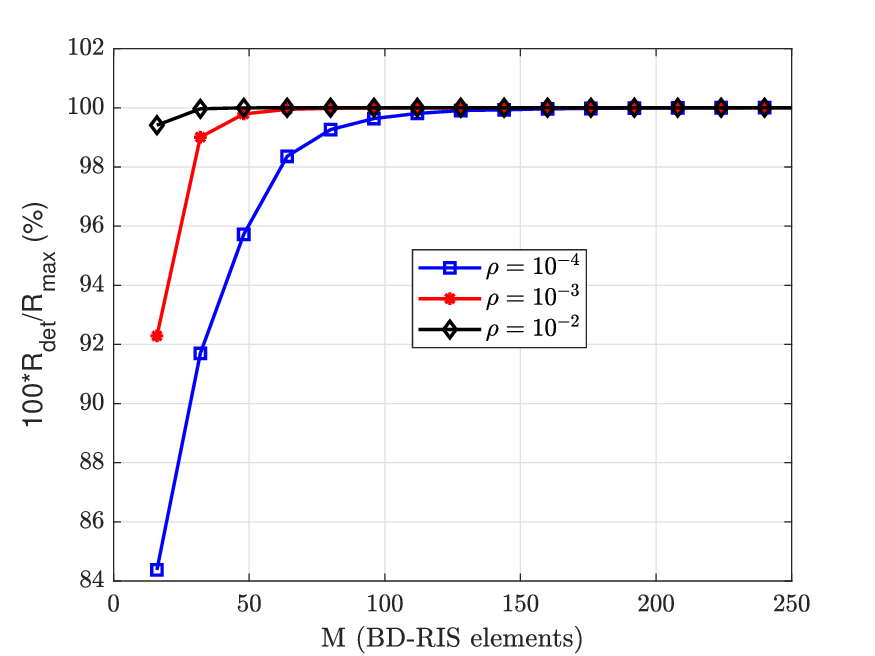}
     \caption{Percentage of the achievable rate obtained by the Max-Det solution as $M$ increases in a BD-RIS-assisted $6\times 3$ MIMO Rayleigh channel.} 
	\label{fig:RatePercentagevsM_6x3}
\end{figure}
\subsection{Optimal Tx covariance matrix versus uniform power allocation}
\begin{figure}[htb]
    \centering
\includegraphics[width=.47\textwidth]{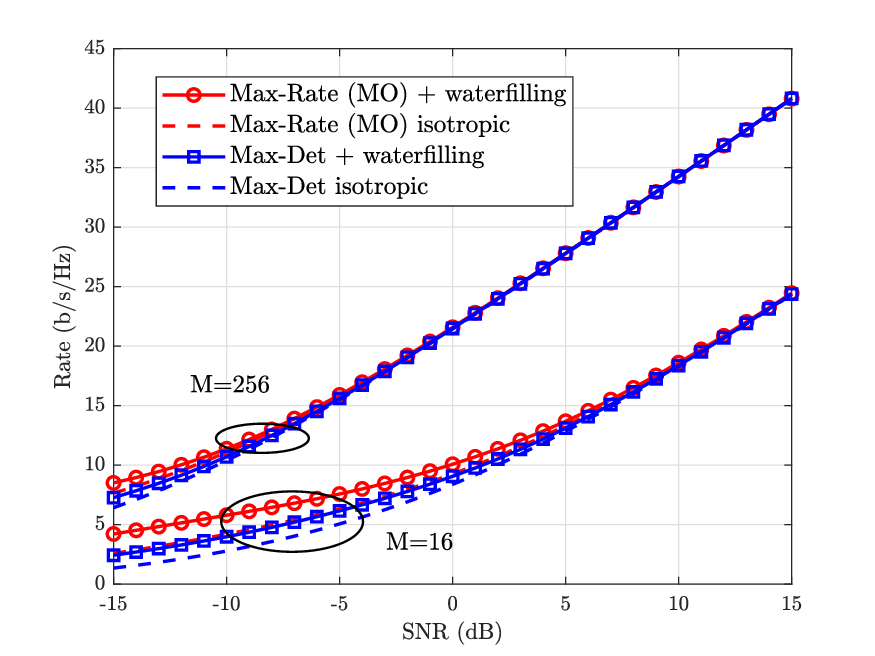}
     \caption{Comparison of isotropic and water-filling-based solutions in a BD-RIS-assisted $8\times 4$ MIMO scenario.} 
	\label{fig:RatevsSNR_8x4_Rxopt}
\end{figure}

We will now analyze the impact of optimizing the transmission covariance matrix versus using isotropic transmissions. As an upper benchmark, we use a joint alternating optimization algorithm for the scattering matrix and the covariance matrix that maximizes the achievable rate \cite{santamariaICASSP26}, \cite{SantamariaUsArxiv}. We consider an $8\times 4$ MIMO system with the scenario parameters described in \ref{sec:scenario}. Fig. \ref{fig:RatevsSNR_8x4_Rxopt} shows the rate versus SNR for the Max-Det solution compared to the solution that maximizes the achievable rate for two values of $M$ and increasing SNRs. As expected, only for low SNR values (which become lower as $M$ increases) does optimizing the transmission covariance matrix offer advantages. The differences between the isotropic and the water-filling solutions are negligible at high and moderate SNRs.

\subsection{Max-Det vs. Max-Rate} 
The objective of this experiment is to compare the solution that maximizes the rate and the proposed solution that maximizes the absolute value of the determinant of the equivalent channel (Algorithm \ref{alg:snippet}). We consider a $4 \times 4$ MIMO scenario with the direct link blocked so communication is only possible through the channel created by the BD-RIS. We compare the results obtained by the following methods:
\begin{itemize}
    \item The unitary BD-RIS that maximizes rate (and also the determinant) (see Remark \ref{remark:unitary}). In the figure it is labeled as ``Max-Rate Unit''.
    \item  The unitary and symmetric BD-RIS that maximizes the determinant (proposed solution). Labeled as `` Max-Det Unit. + Symm.''.
    \item  The unitary and symmetric BD-RIS that maximizes the rate obtained with the MO algorithm proposed in \cite{santamariaICASSP26}. Labeled as `` Max-Rate Unit. + Symm. (MO)''.
    \item A diagonal RIS that maximizes the achievable rate with the algorithm proposed in \cite{ZhangCapacityJSAC2020}. Labeled as `` Max-Rate RIS''.
\end{itemize}

The results in Fig. \ref{fig:RatevsSNR} show that maximizing the determinant is an excellent proxy to maximize rate in a scenario without a direct channel. As can be seen, the differences with the optimal unitary BD-RIS are minimal. Moreover, at high SNRs and when the number of BD-RIS elements increases, the differences are negligible, thus validating the results of Remark \ref{remark:Mgrows}. Furthermore, the unitary and symmetric solution obtained with the MO algorithm in \cite{santamariaICASSP26} provides rates indistinguishable from those of the unitary solution. Finally, a diagonal RIS that maximizes the achievable rate provides a competitive solution for low values of $M$ and low SNRs, but the greater flexibility of a BD-RIS demonstrates its advantages as the number of elements increases.

\begin{figure}[htb]
    \centering
\includegraphics[width=.47\textwidth]{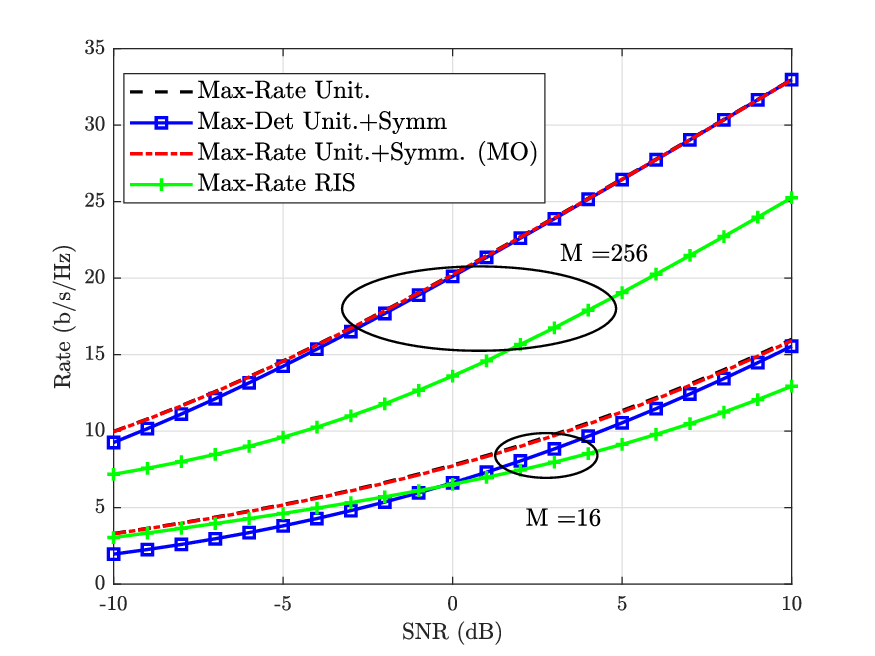}
     \caption{Rate vs. SNR in a $4 \times 4$ MIMO scenario without direct channel for the following solutions: i) the unitary solution that maximizes capacity (dashed red line); ii) the unitary and symmetric solution that maximizes the determinant (solid blue with markers); iii) the unitary and symmetric solution that maximizes capacity using the MO algorithm proposed in \cite{santamariaICASSP26}; and iv) a diagonal RIS (solid green line with markers) that maximizes capacity using the algorithm proposed in \cite{ZhangCapacityJSAC2020}}.
	\label{fig:RatevsSNR}
\end{figure}

To provide a more quantitative evaluation of the reduction in computational complexity of the closed-form Max-Det solution compared to an iterative optimization of the achievable rate, such as that of \cite{santamariaICASSP26}, Fig. \ref{fig:RunTime} shows the runtime of the two methods as $M$ increases for an $8 \times 8$ MIMO system.

\begin{figure}[htb]
    \centering
\includegraphics[width=.47\textwidth]{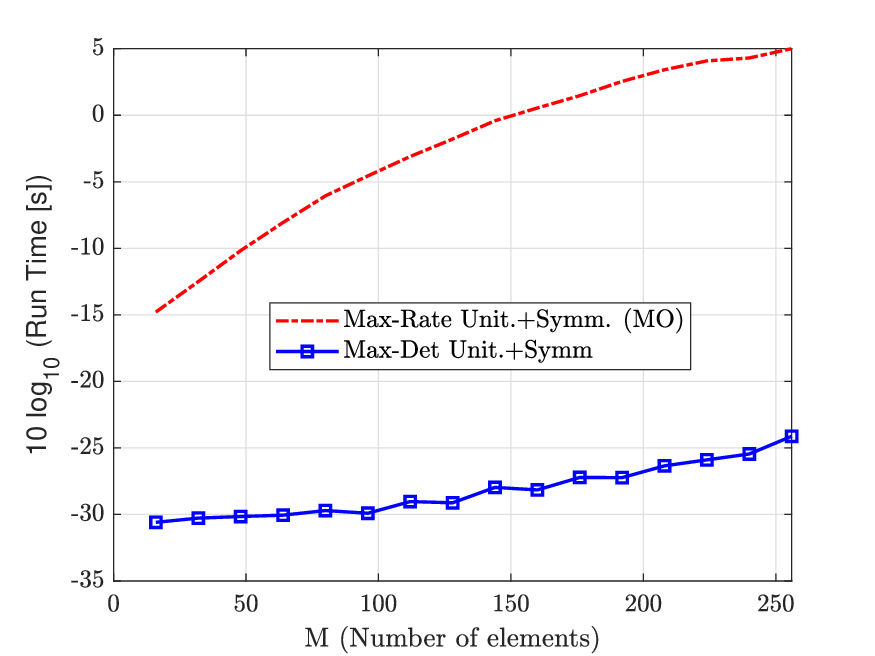}
     \caption{Run time in secs of the MO iterative Max-Rate algorithm in comparison to the Max-Det solution.}
	\label{fig:RunTime}
\end{figure}

\subsection{Unblocked direct channel.}
In this experiment, we consider a scenario with an unblocked direct channel. To control the influence of the direct channel, we generate the channel $\H_d$ as described in Subsection \ref{sec:scenario}, and then scale it by a real value $a \geq 0$ that we vary while keeping the BD-RIS channels $\F$ and $\G$ fixed. The equivalent channel is
\begin{equation*}
    \H(\Thetab) = a \H_d + \F \Thetab \G^H.
\end{equation*}
In the experiment, we consider a $4 \times 4$ MIMO system assisted by a BD-RIS with $M=16$ reflective elements. The SNR at the receiver, before scaling, is 10 dB. All the methods compared in this experiment design symmetric and passive BD-RISs. Specifically, we compare the following methods:
\begin{itemize}
    \item The MO iterative algorithm proposed in \cite{santamariaICASSP26}, which maximizes rate with or without direct channel.
    \item The proposed Max-Det solution without phase correction.
    \item The proposed  Max-Det solution with phase correction. To find the phase term that maximizes capacity (cf. Problem ${\cal P}_3$ in \eqref{eq:phi}), we applied numerical optimization.
    \item The low-complexity closed-form solution proposed in \cite{MaoCL2024}. This suboptimal solution assumes the existence of a direct channel and seeks to maximize the Frobenius norm of the equivalent channel.
\end{itemize}
For completeness, we also include in the comparison the results without BD-RIS and with a random BD-RIS. Fig. \ref{fig:Direct_Link} shows the rate results when the scale factor of the direct link varies between $a=10^{-3}$ and $a=20$. We have highlighted three regions in which different behaviors can be observed. In the first region (green shaded region), where the direct channel is very weak, the Max-Det solutions (with or without phase correction) provide practically identical results to those of the MO algorithm, which is consistent with the results presented in Fig. \ref{fig:RatevsSNR}. The low-cost method \cite{MaoCL2024} degrades somewhat in this region, but not significantly. Clearly, the BD-RIS link determines the performance in this region. In the second region (blue shaded region), the strengths of the direct link and the BD-RIS-link are similar. Here, the phase-corrected Max-Det solution shows some improvement with respect to the standard Max-Det, but there is still a noticeable difference with the MO solution. In this regime, the determinant of the BD-RIS channel is not a good proxy for rate, which is to be expected. Finally, in the third region to the right (red shaded region), the direct channel dominates; therefore, the use of a BD-RIS (or not) does not significantly impact performance.

\begin{figure}[htb]
    \centering
\includegraphics[width=.47\textwidth]{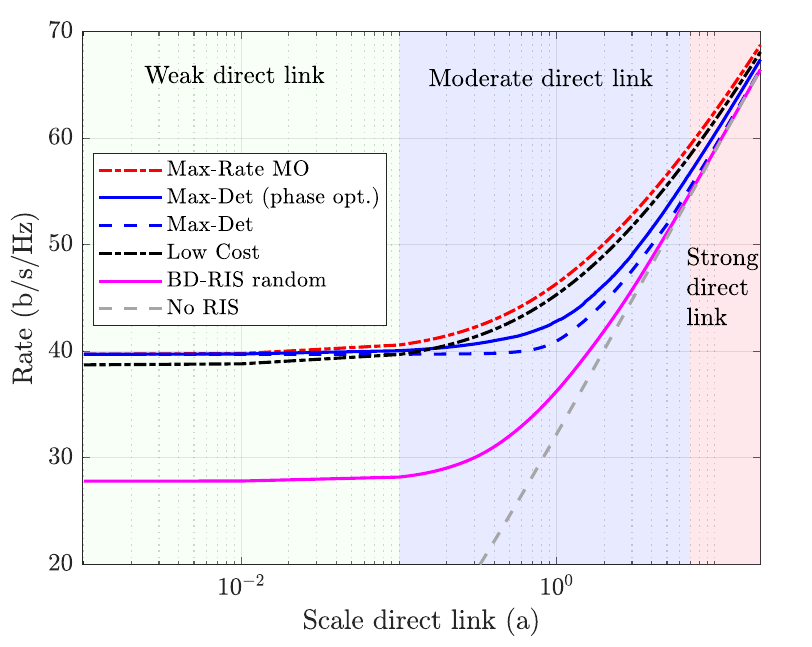}
     \caption{Rate vs. the scaling factor of the direct channel for a $4 \times 4$ MIMO scenario. }
	\label{fig:Direct_Link}
\end{figure}

\subsection{Analysis of $q$-stem implementation}
In this last example, we consider the implementation of the BD-RIS using either a fully connected architecture with $M =64$ elements (providing an upper bound on the achievable rate) or the $q$-stem architecture \cite{ICC_QStem}. To obtain the susceptance matrix $\B$ with the $q$-stem structure, we applied the algorithm described in \cite{Zhou2025Arxiv}. For the fully connected implementation, we simply applied the Cayley transform \cite{Hassibi02}, \cite{Cayley46} to the scattering matrix: $jZ_0\B = (\I_M +\Thetab)^{-1}(\I_M -\Thetab)$. Fig. \ref{fig:qstem} shows the rate achieved by the $q$-stem architecture versus the number of stems $q$. The scenario is a $4\times 4$ MIMO channel; therefore, the Max-Det solution has rank 8. We observe in Fig. \ref{fig:qstem} that, with exactly $q = 2r-1 = 7$ stems, the $q$-stem topology perfectly recovers the optimal $\Thetab$ without any penalty with respect to the fully-connected implementation. 

\begin{figure}[htb]
    \centering
\includegraphics[width=.47\textwidth]{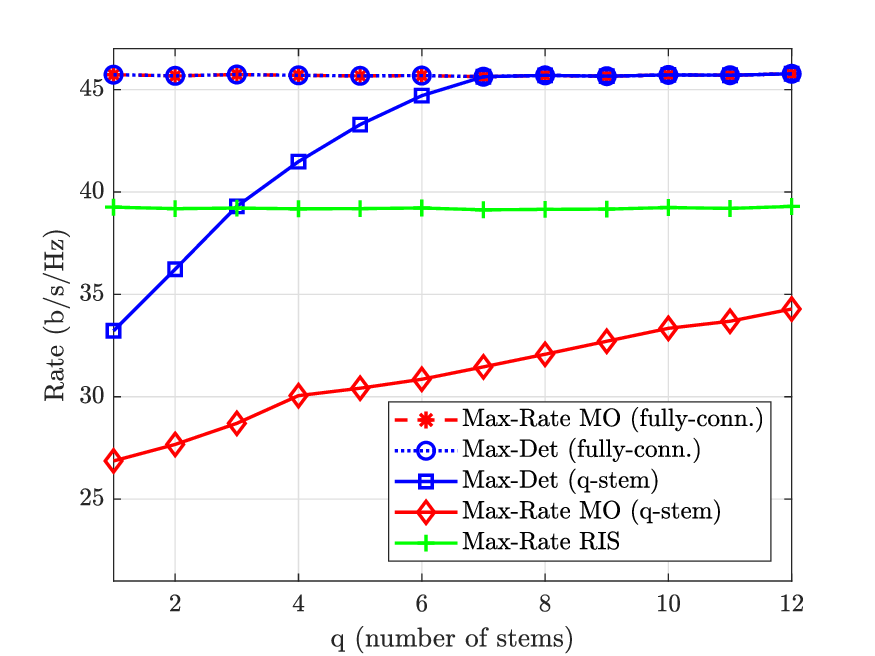}
     \caption{Rate vs. number of stems $q$ for a $q$-stem topology implementation. }
	\label{fig:qstem}
\end{figure}

\section{Conclusions}
\label{sec:conclusions}
Assuming that there is no direct channel, this work has derived a closed-form solution to the problem of maximizing the absolute value of the determinant in a MIMO link assisted by a passive and symmetric BD-RIS. An analysis of the solution allows us to bound the gap between the Max-Det solution and the solution that maximizes rate, showing that the gap decreases as the SNR or the number of BD-RIS elements increases. The rank of the optimal scattering matrix is twice the number of DoFs of the MIMO link, a fact that translates into fewer connection impedances when BD-RIS is implemented in the $q$-stem topology. The numerical simulations show that, when the direct channel is blocked, the rate achieved by the Max-Det solution is very close to that obtained with iterative algorithms that maximize rate, with a computational cost orders of magnitude lower. The solution can be extended to the multiple access channel, and probably to other BD-RIS architectures, such as the band-connected \cite{WuTIT25}. It is also of interest for future research to study whether it is possible to obtain closed-form BD-RIS solutions for other cost functions besides the determinant, such as the Frobenius norm of the equivalent channel or the mean squared error. Investigating the impact of imperfect CSI and the associated estimation overhead on the proposed scheme is an important direction for future work. Finally, analyzing the potential impact of hardware non-idealities in the BD-RIS, such as quantization errors and mutual coupling effects, as well as developing robust solutions to address them, is another area of interest for continuing this work.

\section*{Appendix A: Proof of Theorem \ref{th:theoremU}}
\label{AppendixA}
To simplify the notation of the proof, we shall consider a square $ r\times r $ MIMO channel ($N_t = N_r = r$). The generalization to asymmetric MIMO systems with an unequal number of antennas for transmission and reception is straightforward. Substituting in \eqref{eq:MIMOchanneleq} the SVD factorizations of $\F$ and $\G$ given in \eqref{eq:SVDF} and \eqref{eq:SVDG}, the equivalent channel can be expressed as
\begin{equation*}
 \H(\Thetab)  =  \U_F \begin{bmatrix} \Sigmab_F & {\bf 0}_{r \times (M-r)} \end{bmatrix} \V_F^H \Thetab  \V_G \begin{bmatrix} \Sigmab_G  \\  {\bf 0}_{(M-r) \times r} \end{bmatrix} \U_G^H.
\end{equation*}
The optimal channel that solves ${\cal P}_2$ in \eqref{eq:rank1}-\eqref{eq:roptd} is generated by any $\Thetab$ such that 
\[
|\det(\H(\Thetab))| = \det(\Sigmab_F)\det(\Sigmab_G) = \prod_{i=1}^r \sigma_{f_{[i]}}\sigma_{g_{[i]}} = {\rm D}_{\rm max}.
\]
Certainly, a unitary (but not symmetric) optimal scattering matrix is $\Thetab = \V_F\V_G^H$, and a reduced rank-$r$ version is  $\Thetab = \V_{F_1}\V_{G_1}^H$ (see Remark \ref{remark:unitary}). The question, therefore, is how to construct a symmetric $\Thetab$ such that $|\det(\H(\Thetab))| = \det(\Sigmab_F)\det(\Sigmab_G)$. The answer to this question is that, to maximize the absolute value of the determinant, $\Thetab$ must fulfill the following condition
\begin{equation}
    {\bf T } \!=\!  \underbrace{\begin{bmatrix} \V_{F_1}^H \vspace{0.1cm} \\ \V_{F_2}^H \end{bmatrix}}_{\V_F^H} \Thetab \underbrace{\left [\V_{G_1} \, \V_{G_2} \right]}_{\V_G} \!=\! \begin{bmatrix}{\bf T}_1 & {\bf 0}_{r \times (M-r)} \\ {\bf 0}_{(M-r) \times r} & {\bf T}_2 \end{bmatrix},
    \label{eq:T}
\end{equation}
with ${\bf T}_1 \in \CU(r)$ a unitary $r\times r$ block, i.e., $\T_1^H \T_1 = \I_r$, and $\T_2$ arbitrary. It is trivial to prove that any $\Thetab$ such that the resulting $\T$ has the structure in \eqref{eq:T} satisfies \eqref{eq-20}, shown on the top of the next page, which attains the maximum $\det(\Sigmab_F) \det(\Sigmab_G)$ if $\T_1$ is unitary.
\begin{figure*}
\begin{align}\nonumber
 |\det(\H(\Thetab))| &= \left |\det \left ( \U_F \begin{bmatrix} \Sigmab_F & {\bf 0}_{r \times (M-r)} \end{bmatrix}  \begin{bmatrix}{\bf T}_1 & {\bf 0}_{r \times (M-r)} \\ {\bf 0}_{(M-r) \times r} & {\bf T}_2 \end{bmatrix}  \begin{bmatrix} \Sigmab_G  \\  {\bf 0}_{(M-r) \times r} \end{bmatrix} \U_G^H \right) \right | \\ & 
 =  |\det(\U_F)|\det(\Sigmab_F) |\det(\T_1)| \det(\Sigmab_G) |\det(\U_G)|,\label{eq-20}
\end{align}

    \hrulefill 
\end{figure*}
\setcounter{equation}{20}
Consequently, we need to prove that, with $\Thetab = \Q\Q^T$ chosen as indicated in Theorem \ref{th:theoremU}:
\begin{itemize}
    \item[i)] the resulting $ \T =\V_F^H \Thetab \V_G$ is structured as $\T = \blkdiag({\bf T}_1, {\bf T}_2)$, with ${\bf T}_1$ $r \times r$; and
    \item[ii)] ${\bf T}_1^H \T_1 = \I_r$.
\end{itemize}
We prove each of these claims separately:

{i) \bf $\T$ is block-diagonal with $r \times r$ and $(M-r) \times (M-r)$ blocks.} \\
Assume that the columns of $\V_{F_1}$ and $\V_{G_1}^*$ are linearly independent, so the subspace spanned by the $M \times 2r$ matrix $\A = \left [\V_{F_1} , \V_{G_1}^* \right ]$ is $2r$-dimensional. The compact SVD is $\A = \U \Sigmab \V^H$, where $\U \in \mathbb{S}_t(2r,\mathbb{C}^M)$ is an $M \times 2r$ basis for the columns of $\A$, $\Sigmab = \diag(\sigma_1,\ldots,\sigma_{2r})$, and $\V$ is $2r \times 2r$ unitary. Let us define $\U_1 = \U_{1:r}$ and $\U_2 = \U_{r+1:2r}$, which group the first $r$ and last $r$ columns of $\U$, respectively, and form the matrices ${\bf P}_1 = \U_1 \U_1^T$ and ${\bf P}_2 = \U_2 \U_2^T$. With these definitions, it can be seen that $\Thetab = {\bf P}_1-{\bf P}_2 = \mathbf{U} \D \mathbf{U}^T$ where $\mathbf{D} = \operatorname{blkdiag}(\mathbf{I}_r, -\mathbf{I}_r)$.

To prove that $\T$ has the required block-structured pattern, we need to prove that, for any vector $\g^{\perp}$ such that $\V_{G1}^H \g^{\perp} = \0_{r \times 1}$, then $\V_{F_1}^H ({\bf P}_1 - {\bf P}_2) \g^\perp = \0_{r \times 1}$. Similarly, for any vector $\f^{\perp}$ such that $\V_{F_1}^H \f^{\perp} = \0_{r \times 1}$, then $(\f^{\perp})^H ({\bf P}_1 - {\bf P}_2) \V_{G_1} = \0_{1 \times r}$. The proof only considers the claim $\V_{F_1}^H ({\bf P}_1 - {\bf P}_2) \g^\perp = \0_{r \times 1}$, since the claim $(\f^{\perp})^H ({\bf P}_1 - {\bf P}_2) \V_{G_1} = \0_{1 \times r}$ can be proven similarly.

Define the coordinate matrices $\mathbf{C}_f = \mathbf{U}^H \mathbf{V}_{F_1}$ and $\mathbf{C}_g = \mathbf{U}^T \mathbf{V}_{G_1}$ \(\), both \(2r \times r\). Then $\mathbf{V}_{F_1}^H (\mathbf{P}_1 - \mathbf{P}_2) \g^\perp = \mathbf{C}_f^H \mathbf{D} \, \mathbf{c}_{g}^{\perp}$, where $\mathbf{c}_{g}^{\perp} = \mathbf{U}^T \g^\perp$. The orthogonality condition $\V_{G_1}^H {\g}^{\perp} = \mathbf{0}$ implies $\mathbf{C}_g^H \mathbf{c}_{g}^{\perp} = \mathbf{0}$, and the goal is to prove that $ \mathbf{C}_{f}^H \mathbf{D} \, \mathbf{c}_{g}^{\perp} = \0$.

The Gram matrix, $\mathbf{A}^H \mathbf{A}$, can be factored as
\begin{equation}
\mathbf{A}^H \mathbf{A} = \V \Sigmab^2 \V^H = \begin{bmatrix} \mathbf{I}_r & \boldsymbol{\Gamma} \\ \boldsymbol{\Gamma}^H & \mathbf{I}_r \end{bmatrix},
\label{eq:Gramm}
\end{equation}
where $\boldsymbol{\Gamma} = \mathbf{V}_{F_1}^H \mathbf{V}_{G_1}^*$ is $r\times r$. Let $\boldsymbol{\Gamma} = \mathbf{P} \boldsymbol{\Delta} \mathbf{R}^H$ be the SVD of $\boldsymbol{\Gamma}$, with $\boldsymbol{\Delta} = \operatorname{diag}(\cos(\theta_1), \dots, \cos(\theta_r))$ and $\mathbf{P}, \mathbf{R}$ unitary $r \times r$. The angles $0 \leq \theta_k \leq \pi/2$, $k=1,\ldots,r$ are the principal angles between the subspaces spanned by $\V_{F_1}$ and $\V_{G_1}^*$ \cite[Chapter 9]{Coherence}. The singular values of the Gram matrix $\mathbf{A}^H \mathbf{A}$ in \eqref{eq:Gramm} are $1 \pm \cos(\theta_k)$ for $k=1,\dots,r$. The columns of the unitary matrix $\V$ in \eqref{eq:Gramm} are 
\[
\mathbf{v}_k = \frac{1}{\sqrt{2}} \begin{bmatrix} \mathbf{p}_k \\ \mathbf{r}_k \end{bmatrix},
\]
corresponding to singular values $1 + \cos(\theta_k)$ for $k=1,\ldots,r$, and 
\[
\mathbf{v}_k = \frac{1}{\sqrt{2}} \begin{bmatrix} \mathbf{p}_k \\ -\mathbf{r}_k \end{bmatrix},
\]
corresponding to singular values $1 - \cos(\theta_k)$ for $k=r+1,\ldots, 2r$, where $(\mathbf{p}_k, \mathbf{r}_k)$ are the $k$-th columns of $\mathbf{P}, \mathbf{Q}$. Thus,
\[
\mathbf{V} = \frac{1}{\sqrt{2}} \begin{pmatrix} \mathbf{P} & \mathbf{P} \\ \mathbf{R} & -\mathbf{R} \end{pmatrix}.
\]
With the previous factorization, the matrix with the singular values of $\A$ can be partitioned as $\Sigmab = \operatorname{blkdiag}(\boldsymbol{\Sigma}^+, \boldsymbol{\Sigma}^-)$, where $\boldsymbol{\Sigma}^+ = \operatorname{diag}(\sqrt{1+\cos(\theta_1)}, \dots, \sqrt{1+\cos(\theta_r)})$ and $\boldsymbol{\Sigma}^- = \operatorname{diag}(\sqrt{1-\cos(\theta_1)}, \dots, \sqrt{1-\cos(\theta_r)})$, and $\A^H \A = \V \Sigmab^2 \V^H$.
\begin{figure*}
\setcounter{equation}{23}
\begin{multline}
    \left( \V_{F_1}\V_{F_1}^H + \V_{G_1}^* \V_{G_1}^T\right) \underbrace{ \frac{1}{\sqrt{2 (1 + \cos(\theta_k))}} \left(\V_{F_1} {\bf p}_k + \V_{G_1}^* {\bf r}_k \right)}_{\u_{1,k}} = \\
    \frac{1}{\sqrt{2 (1 + \cos(\theta_k))}} \left( \V_{F_1} {\bf p}_k + \V_{F_1} (\mathbf{P} \boldsymbol{\Delta} \mathbf{R}^H){\bf r}_k + \V_{G_1}^* (\mathbf{R}\boldsymbol{\Delta} \mathbf{P}^H) {\bf p}_k + \V_{G_1}^* {\bf r}_k \right) = \\
    \frac{1}{\sqrt{2 (1 + \cos(\theta_k))}} \left( \V_{F_1} {\bf p}_k + \cos(\theta_k) \V_{F_1} {\bf p}_k + \cos(\theta_k)\V_{G_1}^*{\bf r}_k + \V_{G_1}^* {\bf r}_k \right) = 
    \sqrt{\frac{1 + \cos(\theta_k)}{2}} \left(\V_{F_1} {\bf p}_k + \V_{G_1}^* {\bf r}_k \right) = \sigma_k^2 \u_{1,k},  
    \label{eq-24}
\end{multline}
    \hrulefill 
\end{figure*}
\setcounter{equation}{21}
Le us define the selection matrices $\mathbf{E}_f = \begin{pmatrix} \mathbf{I}_r \\ \mathbf{0}_r \end{pmatrix}$ and $\mathbf{E}_g = \begin{pmatrix} \mathbf{0}_r \\ \mathbf{I}_r \end{pmatrix}$, both \(2r \times r\). Notice that $\V_{F_1} = \A \mathbf{E}_f = \U \Sigmab \V^H \mathbf{E}_f$ and $\V_{G_1}^* = \A \mathbf{E}_g = \U \Sigmab \V^H \mathbf{E}_g$. Then, 
\begin{align*}
 \mathbf{C}_f &= \boldsymbol{\Sigma} \mathbf{V}^H \mathbf{E}_f = \frac{1}{\sqrt{2}} \begin{bmatrix} \boldsymbol{\Sigma}^+ \mathbf{P}^H \\ \boldsymbol{\Sigma}^- \mathbf{P}^H \end{bmatrix},\quad  {\rm and} 
 \\
 \mathbf{C}_g &= \boldsymbol{\Sigma} \mathbf{V}^T \mathbf{E}_g = \frac{1}{\sqrt{2}} \begin{bmatrix} \boldsymbol{\Sigma}^+ \mathbf{R}^T \\ -\boldsymbol{\Sigma}^- \mathbf{R}^T \end{bmatrix}.  
\end{align*}
Remembering that ${\bf P}$ and $\R$ are unitary matrices, from the previous definitions it follows that
\begin{equation}
    \D \C_f = \C_g \R^* {\bf P}^H.
    \label{eq:relationCfCg}
\end{equation}
Finally, we have
\begin{align*}
 \mathbf{C}_f^H \mathbf{D} \, \mathbf{c}_{g}^{\perp} &= \left( \D \C_f \right)^H  \mathbf{c}_{g}^{\perp} \overset{(a)}{= } \left( \C_g ({\bf R}^* {\bf P}^H) \right)^H  \mathbf{c}_{g}^{\perp}\\
 &= ( {\bf P}\R^T) \C_g^H \mathbf{c}_{g}^{\perp} = ({\bf P}\R^T)  \0 = \0,
\end{align*}
where in $(a)$ we have used \eqref{eq:relationCfCg}. Therefore, the SVD basis aligns the coordinates such that \(\mathbf{D} \mathbf{C}_f = \mathbf{C}_g \mathbf{W}\), with $\mathbf{W} = {\bf R}^*{\bf P}^H$ ensuring that $\mathbf{C}_f^H \mathbf{D} \, \mathbf{c}_{g}^{\perp}$ vanishes for any $\mathbf{c}_{g}^{\perp}$ orthogonal to $\operatorname{span}(\mathbf{V}_{G_1})$. This completes the proof of claim i). It remains to prove claim ii).

{ii) \bf $\T_1 = \V_{F_1}^H \Q \Q^T \V_{G_1}$ is unitary.} 

The singular values of $\boldsymbol{\Gamma} = \mathbf{V}_{F_1}^H \mathbf{V}_{G_1}^* =\mathbf{P} \diag(\cos(\theta_1), \ldots, \cos(\theta_r)) \mathbf{R}^H$ are the cosines of the principal angles between the subspaces $\V_{F_1}$ and $\V_{G_1}^*$, while $\R$ and ${\bf P}$ are rotation matrices that allow us to decouple the problem in an interesting way. More specifically, let us consider the $k$-th vector of $\U_1$, which corresponds to the singular value $\sigma_k = \sqrt{1 + \cos(\theta_k)}$ of $\A$. It can be expressed as
\begin{equation}
    \u_{1,k} = \frac{1}{\sqrt{2 (1 + \cos(\theta_k))}} \left(\V_{F_1} {\bf p}_k + \V_{G_1}^* {\bf r}_k \right),
    \label{eq:u1k}
\end{equation}
where ${\bf p}_k$ and ${\bf r}_k$ are, respectively, the $k$-th column of ${\bf P}$ and ${\bf R}$. Now, we check that $ \u_{1,k}$ is an eigenvector of $\A \A^H = \V_{F_1}\V_{F_1}^H + \V_{G_1}^* \V_{G_1}^T$ with eigenvalue $\sigma_k^2$, yielding \eqref{eq-24} on the top of next page, where we have exploited the orthogonality of ${\bf P}$ and ${\bf R}$. The fact that $\| \u_{1,k} \| =1$ can easily be checked. Similarly, the $k$-th vector of $\U_2$, which corresponds to the singular value $\sigma_k = \sqrt{1 - \cos(\theta_k)}$, can be expressed as
\setcounter{equation}{24}
\begin{equation}
    \u_{2,k} = \frac{1}{\sqrt{2 (1 - \cos(\theta_k))}}  \left(\V_{F_1} {\bf p}_k - \V_{G_1}^* {\bf r}_k \right).
    \label{eq:u2k}
\end{equation}
Using this decomposition, we observe that 
\begin{align*}
    \V_{F_1}^H \u_{1,k} &= \frac{1}{\sqrt{2 (1 + \cos(\theta_k))}}   \V_{F_1}^H\left(\V_{F_1} {\bf p}_k + \V_{G_1}^* {\bf r}_k \right) \\
    &= \frac{1}{\sqrt{2 (1 + \cos(\theta_k))}} ({\bf p}_k +  \cos(\theta_k) {\bf p}_k ) \\
    &=\sqrt{\frac{(1 + \cos(\theta_k))}{2}} {\bf p}_k = \cos(\theta_k/2) {\bf p}_k,
\end{align*}
and
\begin{align*}
    \u_{1,k}^T \V_{G_1}  &= \frac{1}{\sqrt{2 (1 + \cos(\theta_k))}} \left(\V_{F_1} {\bf p}_k + \V_{G_1}^* {\bf r}_k \right)^T  \V_{G_1} \\
    &= \frac{1}{\sqrt{2 (1 + \cos(\theta_k))}} ({\bf r}_k^T +  \cos(\theta_k){\bf r}_k^T ) \\
    &=\sqrt{\frac{(1 + \cos(\theta_k))}{2}} {\bf r}_k^T = \cos(\theta_k/2) {\bf r}_k^T,
\end{align*}
so the contribution from each $\V_{F_1}^H\u_{1,k} \u_{1,k}^T \V_{G_1}$ is $\cos(\theta_k/2)^2 {\bf p}_k {\bf r}_k^T$. Similarly, we have
\begin{align*}
    \V_{F_1}^H \u_{2,k} &= \sqrt{\frac{(1 - \cos(\theta_k))}{2}} {\bf p}_k = \sin(\theta_k/2) {\bf p}_k, \\
   \u_{2,k}^T \V_{G_1} &= - \sqrt{\frac{(1 - \cos(\theta_k))}{2}} {\bf r}_k^T = - \sin(\theta_k/2) {\bf r}_k^T
\end{align*}
and hence the contribution from each  $\V_{F_1}^H\u_{2,k} \u_{2,k}^T\V_{G_1}$ is $-\sin(\theta_k/2)^2 {\bf p}_k {\bf r}_k^T$. We remind the reader that the principal angles $\theta_k$ are defined in the interval $0\leq \theta_k \leq \pi/2$, therefore $\cos(\theta_k/2)$ and $\sin(\theta_k/2)$ are positive values. All in all, we have that
\begin{align*}
    \T_1 &= \V_{F_1}^H \Thetab \V_{G_1} = \V_{F_1}^H (\U_1\U_1^T - \U_2 \U_2^T) \V_{G_1}  \\
    &= \sum_{k=1}^r \left( \cos(\theta_k/2)^2 {\bf p}_k {\bf r}_k^T + \sin(\theta_k/2)^2 {\bf p}_k {\bf r}_k^T\right)
    \\&= \sum_{k=1}^r {\bf p}_k {\bf r}_k^T = {\bf P}{\bf \R}^T.
\end{align*}
Finally $\T_1^H \T_1 = {\bf \R}^*{\bf P}^H {\bf P} {\bf R}^T = {\bf \R}^* {\bf \R}^T = \I_r $, thus proving the result.


\addcontentsline{toc}{section}{Appendices}
\renewcommand{\thesubsection}{\Alph{subsection}}
\bibliographystyle{IEEEbib}
\bibliography{refs}


\end{document}